\documentclass[sigconf,screen]{acmart}
\usepackage{longtable} 

\usepackage[shortlabels]{enumitem}
 
\usepackage{url}

 \usepackage[all]{xy}
\usepackage{float}
\usepackage{graphicx}
\usepackage{amsmath}
\usepackage{titlesec}
\usepackage{verbatim}
\usepackage{cmll}
\usepackage{apxproof}
\usepackage{stmaryrd}

\newdir{C}{%
!/4.5pt/@{( }*:(1,-.2)@{}*:(1,+.2)@_{}}
\newdir{|>}{%
!/4.5pt/@{|}*:(1,-.2)@{>}*:(1,+.2)@_{>}}
\newdir{pb}{:(1,-1)@^{|-}}
\def\pb#1{\save[]+<16 pt,0 pt>:a(#1)\ar@{pb{}}[]\restore}

\newcommand{\Sq}{{\rm F}}
\newcommand{\sq}{{\rm f}}

\newcommand{\R}{{\mathcal R}}
\newcommand{\ie}{\textit{i.e.}}
\newcommand{\eg}{\textit{e.g.}}
\newcommand{\cf}{\textit{cf.}}

\newcommand{\viz}{\textit{viz.}}
\newcommand{\red}{\ep_{G,\A}}
\newcommand{\Inst}{\rho_{G,\A}}
\newcommand{\pre}{\,{=\!\!\prec}\,}
\newcommand{\prel}{\,{\prec}\,}

\newcommand{\arr}[1]{{{\stackrel{#1}{\longrightarrow}}}}

\def\mxxth{\mathsurround=0pt}
\dimendef\dimenxx=0
\def\openup{\afterassignment\xxpenup\dimenxx=}
\def\xxpenup{\advance\lineskip\dimenxx
  \advance\baselineskip\dimenxx \advance\lineskiplimit\dimenxx}
\def\eqalign#1{\,\vcenter{\openup1\jot \mxxth
  \ialign{\strut\hfil$\displaystyle{##}$&$\displaystyle{{}##}$\hfil
     \crcr#1\crcr}}\,}

\newif\ifdtxxp
\def\displxxy{\global\dtxxptrue \openup1\jot \mxxth
  \everycr{\noalign{\ifdtxxp \global\dtxxpfalse
      \vskip-\lineskiplimit \vskip\normallineskiplimit
      \else \penalty\interdisplaylinepenalty \fi}}}
\def\displaylines#1{\displxxy
  \halign{\hbox to\displaywidth{$\hfil\displaystyle##\hfil$}\crcr
      #1\crcr}}

\newskip\mycntring \mycntring=0pt plus 1000pt minus 1000pt

\def\leqalignno#1{\displxxy \tabskip=\mycntring
  \halign to\displaywidth{\hfil$\displaystyle{##}$\tabskip=0pt
      &$\displaystyle{{}##}$\hfil\tabskip=\mycntring
      &\kern-\displaywidth\rlap{$##$}\tabskip=\displaywidth\crcr
      #1\crcr}}

\makeatletter
\DeclareFontEncoding{LS1}{}{}
\makeatother
\DeclareFontSubstitution{LS1}{stix}{m}{n}
\DeclareSymbolFont{operatorsstix}{LS1}{stix}{m}{n}
\DeclareSymbolFont{arrows1stix}{LS1}{stixsf}{m}{n}
\DeclareSymbolFont{arrows2stix}{LS1}{stixsf}{m}{it}
\DeclareMathSymbol{\upand}{\mathbin}{operatorsstix}{"C4}
\DeclareMathSymbol{\nvrightarrow}             {\mathrel}{arrows1stix}{"F6}
\DeclareMathSymbol{\lefttail}                 {\mathrel}{arrows2stix}{"B2}

\newcommand{\tri}{\trianglelefteq}
\newcommand{\DA}{{D({\A})}}
\newcommand{\DB}{{D({\B})}}
\newcommand{\DAB}{{D(\A, \B)}}

\newcommand{\Rel}{{\rm Rel}}
\newcommand{\pto}{\rightharpoonup}
\renewcommand{\mod}[1]{| #1 |}
\renewcommand{\phi}{\varphi}
\newcommand{\den}[1]{\llbracket #1 \rrbracket}
\newcommand{\last}{{\rm last}}
\newcommand{\ep}{\epsilon}
\newcommand{\esswp}{esps}
\newcommand{\EX}{{\bf  E}}
\def\profto{\!\!\!\xymatrix@C-.75pc{\ar[r]|-{\! +\!} &}\!\!\! }
\newcommand{\bel}{\sqsubseteq}
\newcommand{\Var}{\mathrm{vc}}
\def\hash{{\#}}

\newcommand{\rstd}{\mathbin\upharpoonright}
 
 \newcommand{\Strat}{{\mathbf{Strat}}}
\newcommand{\all}{\forall}

\newcommand{\A}{{\mathcal A}}
\newcommand{\B}{{\mathcal B}}
\newcommand{\C}{{\mathcal C}}
\newcommand{\D}{{\mathcal D}}

\newcommand{\id}{\textrm{id}}
\newcommand{\pol}{\textit{pol}}

\newcommand{\opmove}{{\boxminus}}
\newcommand{\plmove}{{\boxplus}}

\newcommand{\sncirc}{\oast}
\newcommand{\co}{\mathbin{\textit{co}}}
\newcommand{\vvbar}{{\mathbin{\parallel}}}
\newcommand{\scirc}{{{\odot}}}
\newcommand{\cov}{{{\mathrel-\joinrel\subset}}}
\newcommand{\longcov}[1]{{\stackrel{#1}{\mathrel-\joinrel\relbar\joinrel\subset\,}}}

\newcommand{\imc}{\rightarrowtriangle}
\newcommand{\setdif}{\setminus}
\newcommand{\sig}{\sigma}

\newcommand{\CC}{{\rm C\!\!C}}
\newcommand{\cc}{c\!c}
\newcommand{\conf}[1]{\:\!{\mathcal C}(#1)^o}
\newcommand{\iconf}[1]{\:\!{\mathcal C}(#1)}
\newcommand{\parrow}{\rightharpoonup}
\newcommand{\set}[2]{{\{  #1\  | \  #2 \} }}
\newcommand{\setof}[1]{{\{ #1 \} }}
\newcommand{\eqdef}{\coloneqq}
\newcommand{\iso}{\cong}

\newcommand{\expn}{{ expn}}
\newcommand{\DS}{{\mathbf{SD}}}
\newcommand{\SD}{{\mathbf{SD}}}
\newcommand{\RED}{{\bf Red}}
\newcommand{\GAMESIG}{{\bf Sig
}}

\newcommand{\al}{\alpha}
\newcommand{\be}{\beta}
\newcommand{\ga}{\gamma}

\usepackage{amsthm}
\theoremstyle{plain}
\newtheoremrep{theorem}{Theorem}[section]
\newtheoremrep{lemma}[theorem]{Lemma}
\newtheoremrep{proposition}[theorem]{Proposition}
\newtheoremrep{corollary}[theorem]{Corollary}

\title{Concurrent Games over Relational Structures:\\The Origin of Game Comonads} 
\author{Yo\`av Montacute}
\orcid{1234-5678-9012}
\affiliation{%
  \institution{University of Cambridge}
  \country{United Kingdom}
}

\author{Glynn Winskel}
\orcid{0000-0002-5069-2303}
\affiliation{%
  \institution{University of Strathclyde}
  \country{United Kingdom}}
\email{}

\begin{abstract}
Spoiler-Duplicator games are used in finite model theory to examine the expressive power of logics.
Their strategies have recently been reformulated as coKleisli maps of game comonads over relational structures, providing 
new results in finite model theory via categorical techniques.
We present a novel framework for studying Spoiler-Duplicator games by viewing them as event structures.
We introduce a first systematic method for constructing comonads for all one-sided Spoiler-Duplicator games: 
game comonads are now realised by adjunctions to a category of games, generically constructed from a comonad in a bicategory of game schema (called signature games). 
Maps of the constructed categories of games are strategies and generalise coKleisli maps of game comonads;  in the case of one-sided games they are shown to coincide with suitably generalised homomorphisms.  Finally, we provide characterisations of strategies on two-sided Spoiler-Duplicator games; in a common special case they coincide with spans of event structures.
\end{abstract}
\begin{document}
\maketitle
\section{Introduction}
 Spoiler-Duplicator games are used in finite model theory to study the expressive power of logical languages. 
These games make use of resources (\eg~number of rounds) which correspond to different restrictions on fragments of first-order logic (\eg~quantifier rank). One may venture beyond Spoiler-Duplicator games to a framework that supports both resources and structure. In the semantic world, concurrent games and strategies based on event structures have been demonstrated to achieve that goal: 
they can be composed and also support quantitative extensions, \eg~to probabilistic and quantum computation~\cite{Probstrats,hugo-thesis,concquantumstrats,POPL19,POPL20}, as well as resource usage~\cite{aurore}. 
As a result, concurrent games and strategies 
provide a rich arena in which structure meets power---a line of 
research that became prominent with the paper of Abramsky, Dawar and Wang \cite{AbramskyDawarWang} on the pebbling comonad in finite model theory.

That seminal work has been followed up by many more examples of ``game comonads,''  for example~\cite{AbramskyDawarWang,relatingstructure,Lovasz,pebblerelationgame}. 
Though obtaining the comonads has been something of a mystery, largely because algebraic structures do not in themselves express the operational features of games. 
To some extent this has been addressed through \emph{arboreal categories}~\cite{arboreal,lineararboreal} which attempt to axiomatise the Eilenberg-Moore coalgebras of the known examples; as their name suggests arboreal categories and their ``covers'' impose a treelike behaviour on relational structures so making them amenable to operational concerns.
But, as Bertrand Russell joked, postulation has ``the advantages of theft over honest toil''~\cite{Russell}. And,
the coalgebras of newer comonads such as the all-in-one pebble-relation comonad \cite{pebblerelationgame} are not arboreal and consequently 
the original definition is being readdressed in \emph{linear arboreal categories}~\cite{lineararboreal}.

Our contribution is to provide  a sweeping definition of Spoiler-Duplicator games out of which game comonads, both 
old  and new, their coalgebras and characterisations are derived as general theorems.  To do so,
we introduce concurrent games and strategies over relational structures.
Through their foundation in event structures, concurrent games and strategies represent closely the operational nature of games, their interactivity, dependence, independence, and conflict of moves.
This makes the usual, largely informal and implicit description of Spoiler-Duplicator games formal and explicit. 
In particular, the concurrency expressible in event structures plays an essential role in enforcing the independence of moves.  

A game will be represented by an event structure in which each event stands for a move occurrence of Player (Duplicator) or Opponent (Spoiler); a move is associated with a constant or variable. 
Positions of the game are represented by configurations of the event structure. 
With respect to a many-sorted relational structure $\mathcal A$,
a strategy (for Player) assigns values in $\mathcal A$ to Player variables  (those associated with Player moves) in response to \emph{challenges} or assignments of values in $\mathcal A$ to variables by Opponent.  
A winning condition specifies those configurations, with latest assignments, which represent a win for Player.  
A strategy is winning if it results in a winning configuration of the game regardless of Opponent's strategy.  

We exploit that such games form a bicategory 
%
to exhibit  
a deconstruction of traditional Spoiler-Duplicator games;   the  choice of Spoiler-Duplicator game is parameterised by a comonad $\delta$ in the bicategory--- $\delta$ specifies  the allowed pattern of interaction between Spoiler and Duplicator.  Special cases follow from particular choices of $\delta$.  They include 
 Ehrenfeucht-Fra\"iss\'e games, 
pebbling games, such as the $k$-pebble game~\cite{AbramskyDawarWang} and the all-in-one $k$-pebble game~\cite{pebblerelationgame};
and their versions on transition systems, \ie~simulation and trace inclusion. 
 
With respect to a choice of $\delta$, we obtain a category of strategies on a Spoiler-Duplicator game $\SD_\delta$; and, 
for suitable one-sided games, a comonad $\Rel_\delta(\_)$ on the category of appropriate relational structures. 
For general $\delta$ we provide characterisations of the strategies $\SD_\delta$. In the one-sided case, the strategies $\SD_\delta$ correspond to $\delta$-homomorphisms---strictly more general than coKleisli maps of the comonad $\Rel_\delta(\_)$; we delineate when $\SD_\delta$ coincides with the coKleisli category. We characterise the Eilenberg-Moore coalgebras of $\Rel_\delta(\_)$ as relational structures which also possess a certain event-structure shape---the coalgebras are not always (linear) arboreal categories. 
The derived comonads on relational structures coincide with those in the literature we know, particularly for the examples mentioned.  
This provides a 
systematic method for constructing comonads for one-sided Spoiler-Duplicator games, one in which 
game comonads are now realised by adjunctions to a category of games; this
reconciles the composition of strategies as coKleisli maps, prevalent in the category-theoretic approach to finite model theory, 
 with the more standard composition of strategies following the paradigm of Conway and Joyal~\cite{conway,joyal}. 

\section{Preliminaries}\label{sec:pre}
We cover the basic definitions of games 
as event structures 
\cite{lics11}.
\setcounter{subsection}{1}
\subsubsection{Event structures.}
An \emph{event structure}  is a triple $E= (\mod E, \leq, \hash)$, where $\mod E$ is a set of  \emph{events}, $\leq$ is a partial order of \emph{causal dependency} on $E$, and $\hash$ is a binary irreflexive symmetric relation, called the \emph{conflict relation}, such that $ \set{e'}{e'\leq e}$ is finite and
$e\hash e'{\leq} e'' $ implies $e\hash e''$.
Event structures have an accompanying notion of a state or history:
a {\em configuration} is a (possibly infinite) subset $x\subseteq \mod E$ which is \emph{consistent}, \ie~$\forall e, e'\in x$. $(e,e')\notin\#$, and \emph{down-closed}, \ie~$ e'\leq e\in x  \hbox{ implies } e' \in x$. 
 We denote by $\iconf E$ the set of configurations of $E$ and by $\conf E$ the subset of finite configurations.  
 
 Two events $e$, $e'$ are called \emph{concurrent} when they are neither in conflict nor causally related; we denote this by $e\co e'$.  
 In games, the relation $e\imc e'$ of {\em immediate dependency}, meaning that $e$ and $e'$ are distinct with $e\leq e'$ and no events are in between, plays an important role. 
  We write $[X]$ for the down-closure of a subset of events $X$; when $X$ is a singleton $\setof e$, we denote its down-closure by $[e]$. 
  Note that $[e]$ is necessarily a configuration.  
 Two events $e$ and $e'$ are in {\em immediate conflict} whenever $e\hash e'$ and if $e{\geq} e_1\hash e_1'{\leq} e'$ then $e=e_1$ and $e'=e_1'$. 
To avoid ambiguity, we sometimes distinguish the precise event structure with which a relation is associated and write, for instance, $\leq_E$,  $\imc_E$, $\hash_E$ and  $\co_E$.

In diagrams, events are depicted as squares, immediate causal dependencies by arrows and immediate conflicts by squiggly lines. For example, the diagram in Figure \ref{eventstructure}
represents an event structure with five events.
 The event to the far-right is in immediate conflict with one event---as shown, but in non-immediate conflict with all events besides the one on the lower far-left, with which it is concurrent. 

 \begin{figure}
    \centering
     \includegraphics[scale=0.30]{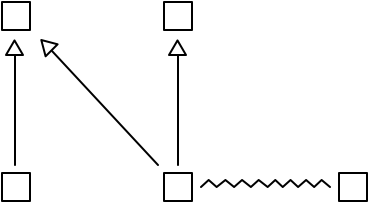}    \caption{An event structure}
     \label{eventstructure}
\end{figure}
Let $E$ and $E'$ be event structures. 
A {\em  map of event structures} $f:E\rightarrow E'$ is a partial function $f:\mod E\parrow \mod{E'}$ on events such that, for all $x\in\iconf E$, we have $fx\in\iconf{E'}$, and if $e, e' \in x$ and $f(e)=f(e')$, with both defined, then $e=e'$. 
Maps of event structures compose as partial functions.  
Notice that  for a total map $f$, the condition on maps says that $f$ is {\em  locally injective}, in the sense that w.r.t.~every configuration $x$ of the domain, the restriction of $f$ to a function from $x$ is injective; the restriction 
to a function from $x$ to $fx$ is thus bijective.

Although a map $f:E\to E'$ of event structures does not generally preserve causal dependency, it does reflect   
causal dependency  locally: whenever $e, e'\in x\in \iconf E$ and $f(e')\leq f(e)$, with $f(e)$ and $f(e')$  both defined, then $e'\leq e$.  
Consequently,
$f$  preserves the concurrency relation: if $e\co e'$
then $f(e)\co f(e')$, when defined.
A total map of event structures is called {\em rigid} when it preserves causal dependency.

\begin{proposition}[\cite{Winskel07}]
A total map $f:E\to E'$ of event structures is rigid iff  for all   $x\in\iconf E$ and $y\in\iconf{E'}$, if $y\subseteq f x$, then there exists $z\in\iconf{E}$ such that $z\subseteq x$ and $f z = y$.
Moreover, $z$ is necessarily unique. 
\end{proposition}

Open maps give a general treatment of bisimulation; they are defined by a path-lifting property~\cite{JNW}. 
Here a characterisation of open maps of event structures will suffice. 

\begin{proposition}[open map~\cite{JNW}]\label{prop:openmap}{
A total map $f:E\to E'$ of event structures is {\em open} iff it is rigid and  for all   $x\in\iconf E$ and $y\in\iconf{E'}$,  $f x \subseteq y$ implies that there exists $ z\in\iconf{E}$ such that $x\subseteq z$ and $f z = y$.}
\end{proposition}

Event structures possess a hiding operation. 
Let $E$ be an event structure and  $V$ a subset of its events.  
The \emph{projection} of $E$ 
to $V$, written 
$E{\mathbin\downarrow}V$,  is defined to be the event structure $(V, \leq_V, \hash_V)$ in which the relations of causal dependency and conflict are simply restrictions of those in $E$ [appendix \ref{app:concomp}].

\setcounter{subsection}{2}
\subsubsection{Concurrent games and strategies.}\label{sec:games-strats}
The driving idea of concurrent games is to replace the traditional role of game trees by that of event structures~\cite{lics11}. 
Both games and strategies will be represented by an
  {\em event structure with polarity (esp)}, which  comprises $(A,\pol_A)$, where $A$ is an event
structure and $\pol_A:A\to \{+,-,0\}$ is a polarity function  ascribing a
polarity $+$ to Player, polarity $-$ to opponent, and polarity $0$ to neutral events.
The events of $A$ correspond to (occurrences of) moves. 
Events of neutral polarity arise in a play between a strategy and a counter-strategy.
Maps between 
esps are those of event structures that preserve polarity. 
A {\em game} is represented by an event structure with polarities restricted to + or $-$, with no neutral events. 

In an event structure with polarity, for
configurations $x$ and $y$, we write $x \subseteq^-y$ (resp.~$x \subseteq^+y$) to mean inclusion in which all the intervening events $(y\setdif x)$ are Opponent (resp.~Player) moves. 
For a subset of events $X$ we write $X^+$ and $X^-$ for its restriction to Player and Opponent moves, respectively. 
There are two fundamentally important operations on games:
Given a game $A$, the \emph{dual game} $A^\perp$ is the same as $A$ but with the polarities reversed. 
The other operation, the {\em simple parallel composition} $A\vvbar B$, is achieved by simply juxtaposing $A$ and $B$; ensuring that two events are in conflict only if they are in conflict in a component. 
Any configuration $x$ of $A\vvbar B$ decomposes into $x_A\vvbar x_B$, where $x_A$ and $x_B$ are configurations of $A$ and $B$, respectively.
\begin{definition}[strategy]{\rm 
A {\em strategy} in a game $A$ is an esp $S$ together with a total map $\sig: S \to A$ of \esswp, where
 \begin{itemize}
 \item 
  if $\sigma x  \subseteq^- y$, for $x\in\iconf S$ and  
  $y\in\iconf A$, then there exists a unique $
 x'\in\iconf S$ such that $x\subseteq x'$ and $\sig x' = y$;  
 \item 
  if 
$s\imc_S s'$ \& ($\pol(s) = +$ or $\pol(s') = -$), then 
$ \sigma(s)\imc_A \sigma(s')$.
 \end{itemize}
 The strategy is {\em deterministic} iff all immediate conflict in $S$ is between Opponent events. We say the strategy is {\em rigid/open} according as the map $\sig$ is rigid/open. }
 \end{definition}

The first condition  is called {\em receptivity}.
It ensures that the strategy is open to all moves of Opponent 
permitted by the game. 
The second condition, called {\em innocence} in~\cite{FP}, ensures that the only additional immediate causal dependencies a strategy can enforce beyond those of the game are those in which a Player move causally depends on Opponent moves.  
An important feature of strategies is that they admit composition to form a bicategory \cite{lics11} [appendix \ref{app:concomp}].
A \emph{map of strategies} $f:\sig \Rightarrow \sig'$, where  $\sig:S\to A$ and $\sig' :S' \to A$, is a map $f:S\to S'$ 
such that $\sig=\sig'f$; this determines when strategies are isomorphic. 
Note that such a map $f$ of strategies is itself receptive and innocent and so a strategy in $S'$.

Following Conway and Joyal~\cite{conway,joyal}, we define a {\em strategy from a game $A$ to a game $B$} as a strategy in the game $A^\perp \vvbar B$. 
The conditions of receptivity and innocence 
precisely 
ensure that  copycat strategies as follow 
behave as 
identities w.r.t.~composition \cite{lics11}. 
%
Given a game $A$, the \emph{copycat strategy} $\cc_A:\CC_A\to A^\perp\vvbar A$ is an instance of a strategy from $A$ to $A$. 
The event structure $\CC_A$ is based on the idea that Player moves in one component of the game $A^\perp\vvbar A$ always copy 
corresponding moves of
Opponent in the other component. For $c \in A^\perp\vvbar A$ we use $\bar c$ to mean the corresponding copy of $c$, of opposite polarity, in the
alternative component. The event structure $\CC_A$ comprises   $A^\perp\vvbar A$ and its causal dependencies with extra causal dependencies 
$\bar c\leq c$ for all  
events $c$ such that $\pol_{A^\perp\vvbar A}(c) = +$; 
this generates a partial order. 
Two events in $\CC_A$ are in conflict if they now causally depend on events originally in conflict in $A^\perp\vvbar A$. Figure~\ref{fig:copycat} illustrates $\CC_A$ when $A$ is $\opmove\imc\plmove$.
The map $\cc_A$ acts as the identity function  on events. 
The  copycat strategy on a game $A$ is deterministic iff $A$ is {\em race-free}, \ie~there is no immediate conflict between a Player and an Opponent move~\cite{DBLP:journals/fac/Winskel12}.
\begin{figure}[H]
\centering
\includegraphics[scale=0.30]{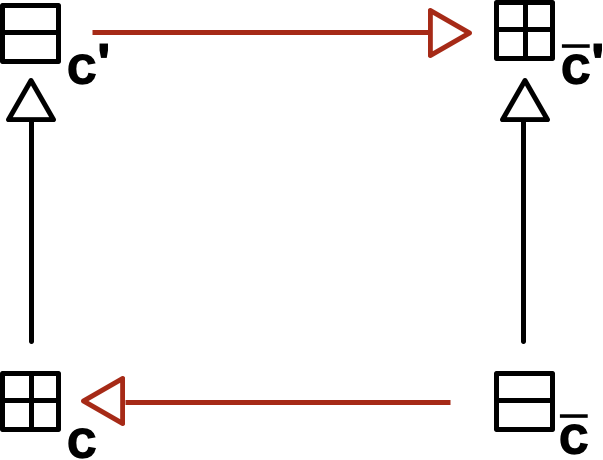}
    \caption{A copycat strategy}
    \label{fig:copycat}
\end{figure}
%
 
\begin{toappendix}
\subsection{Composition in concurrent games}\label{app:concomp} 
\subsubsection{Hiding---the defined part of a map.}
Let $E$ be an event structure.  Let $V\subseteq \mod E$ be a subset events we call \emph{visible}.
Define the \emph{projection} 
on $V$, by 
$E{\mathbin\downarrow} V\eqdef
(V, \leq_V, \hash_V)$, where $v \leq_V v' \hbox{ iff } v\leq v' \ \&\ v,v'\in V$ and $v \hash_V v' \hbox{ iff } v\hash_E v'$ and $v,v'\in V$. 
The operation $\downarrow$  {\em hides} all events in $E\backslash V$. 
It is associated with a {\em partial-total factorisation system}. 
Consider a map of event structures $f:E\to E'$ and let 
$$V\eqdef \set{e\in E}{ f(e) \hbox{ is defined}}.$$
Then $f$ clearly factors into the composition 
$$
\xymatrix{
E\ar@^{->}[r]^{f_0}& E{\mathbin\downarrow} V \ar[r]^{f_1}& E'\\}
$$
of $f_0$, which is a partial map of event structures taking $e\in \mod E$ to itself if $e\in V$ and undefined otherwise, and $f_1$, which is a total map of event structures acting like $f$ on $V$. 
Each $x\in \iconf{E{\mathbin\downarrow} V}$ is the image under $f_0$ of a {\em minimum configuration}, \viz~$[x]_E\in\iconf E$. 
We call $f_0$ a {\em projection} and $f_1$ the {\em defined part} of the 
map 
$f$.    

\subsubsection{Pullbacks.}
The {\em pullback} of total maps of event structures is essential in composing strategies. 
We can build it out of computation paths.  
 A computation path  is described by  a partial order $(p,\leq_p)$  for which the set $\set{e'\in p}{e'\leq_p e}$ is finite for all $e\in p$.  We can identify  such a path with an event structure with no conflict.  
  Between two paths $p=(p, \leq_p)$ and $q=(q, \leq_q)$, we write  $p \hookrightarrow  q$ when $p\subseteq q$ and the  inclusion  is a rigid map of event structures.   

\begin{proposition}[\cite{ecsym-notes}]\label{prop:rigfam}
A {\em rigid family} $\mathcal R$ comprises a non-empty subset of finite partial orders which is down-closed w.r.t.~rigid inclusion, \ie~$p\hookrightarrow q\in {\mathcal R}$ implies $p\in {\mathcal R}$. It is {\em coherent} when a  pairwise-compatible finite subfamily is compatible.   
A coherent rigid family determines 
an event structure $\Pr({\mathcal R})$ whose order of finite configurations is isomorphic to $({\mathcal R}, \hookrightarrow)$.  The event structure $\Pr({\mathcal R})$
has events those elements of $\mathcal R$ with a top event; its causal dependency is given by rigid inclusion; and its conflict by incompatibility w.r.t.~rigid inclusion. The order isomorphism  $\mathcal R\iso \conf{\Pr({\mathcal R})}$   takes $q\in \mathcal R$ to 
$\set{p\in\Pr({\mathcal R})}{p\hookrightarrow q}$.
\end{proposition}

We can define pullback via a rigid family of {\em secured bijections}. Let $\sig:S\to B$ and $\tau:T\to B$ be total maps of event structures. There is a composite bijection 
$$
\theta: x\iso \sig x = \tau y \iso y,
$$
between $x\in\conf S$ and 
 $y\in \conf T$ such that $\sig x = \tau y$; because $\sig$ and $\tau$ are total they induce bijections between configurations and their image. The bijection is {\em secured}  when
the transitive relation generated on $\theta$ by  $(s, t) \leq (s', t')$ if $s\leq_S s'$ or $t\leq_T t'$ is a partial  order.

\begin{theorem}[\cite{ecsym-notes}]
Let $\sig:S\to B$ and $\tau:T\to B$ be total maps of event structures.  The family $\mathcal R$ of secured bijections between $x\in\conf S$ and 
 $y\in \conf T$ such that $\sig x = \tau y$ is a rigid family.  
 The functions $\pi_1: \Pr({\mathcal R})\to S$ and $\pi_2:\Pr({\mathcal R})\to T$, taking a secured bijection with top to, respectively, the left and right components of its top, are maps of event structures.
$ \Pr({\mathcal R})$ with $\pi_1$ and $\pi_2$ is the pullback of $\sig$ and $\tau$ in the category of event structures.
 \end{theorem}
 
W.r.t.~$\sig:S\to B$ and $\tau:T\to B$,  define $x\wedge y$ to be the configuration of their pullback which corresponds via this isomorphism to a secured bijection between $x\in\iconf S$ and $y\in \iconf T$, necessarily with $\sig x = \tau y$; any configuration of the pullback takes the form  $x\wedge y$  for unique $x$ and $y$.  
\subsubsection{Composition of strategies.}

Two strategies $\sig:S\to A^\perp\vvbar B$ and $\tau:T\to B^\perp \vvbar C$
compose via pullback and hiding, summarised below.
\[\vcenter{\small\xymatrix@R=12pt@C=10pt{
 &\ar[dl]_{\pi_1
 }T\sncirc S \ar@^{-->}[rr]
\ar@{..>}[dd]|{\tau\sncirc\sig}
\pb{270}\ar[dr]^{\pi_2
 }&  & T\scirc S \ar@{-->}[dd]^{\tau\scirc \sig}\\
S\vvbar C\ar[dr]_{\sig\vvbar C}&&\ar[dl]^{A\vvbar \tau}A\vvbar T\\
 &A\vvbar B\vvbar C\ar@^{->}[rr]&
 &A\vvbar C}}.\]
Ignoring polarities, by forming the pullback of $\sig\vvbar C$ and $A\vvbar \tau$ we obtain the synchronisation of
complementary moves of $S$ and $T$ over the common game $B$;  subject to the causal constraints of $S$ and $T$, the effect is to instantiate the Opponent moves of $T$ in $B^\perp$ by the corresponding Player moves of $S$ in $B$, and {\it vice versa}.  Reinstating polarities we obtain the {\em interaction}  of $\sig$ and $\tau$  
\[
\tau\sncirc \sig: T\sncirc S \to A^\perp \vvbar B^0 \vvbar C,
\]
where we assign neutral polarities to all moves in or over $B$. Neutral moves over the common game $B$ remain unhidden. The map $A^\perp\vvbar B^0\vvbar C\pto A^\perp\vvbar C$  is undefined on $B^0$ and otherwise mimics the identity. Pre-composing this map with $\tau\sncirc \sig$ we obtain a partial map $T\sncirc S\pto A^\perp\vvbar C$; it is undefined on precisely the neutral events of $T\sncirc S$.  The defined parts of its partial-total factorisation yields 
 \[\tau\scirc \sig: T\scirc S\to A^\perp\vvbar C.\]
 On reinstating polarities; this is the {\em composition} of $\sig $ and $\tau$. 
 We obtain a bicategory $\Strat$   where the objects are games, arrows are strategies, and 2-cells are maps between strategies \cite{lics11,lics12}. Identities are given by copycat strategies.
 
 It is useful to introduce notation for configurations of the interaction and composition of strategies $\sig$ and $\tau$. For $x\in\iconf S$ and $y\in\iconf T$, 
let $\sig x = x_A\vvbar x_B$ and $\tau y = y_B\vvbar y_C$ where $x_A\in \iconf A$, $x_B, y_B\in\iconf B$, $y_C\in\iconf C$.  
Define
$
y\sncirc x = (x\vvbar y_C) \wedge (x_A\vvbar y)
$.
This is a partial operation only  defined if the $\wedge$-expression is. 
It is defined and glues configurations $x$ and $y$ together at their common overlap over $B$ provided $x_B=y_B$ and a finitary partial order of causal dependency results. 
Any 
 configuration of $T\sncirc S$ has the form 
$
y\sncirc x
$,
for unique $x\in\iconf S, y\in\iconf T$.   
Accordingly, any 
 finite configuration of $T\scirc S$ is given as
\[
y\scirc x =   (y\sncirc x)\cap\set{e\in T\sncirc S}{\pol_{T\sncirc S}(e)\neq 0},
\]
for some $x\in\conf S$ and $y\in\conf T$. The configurations $x$ and $y$ need not be unique.  However, w.r.t.~any configuration $z$ of $T\scirc S$ there are {\em minimum} $x,y$ such that $y\scirc x =z$, \viz~those for which $y\sncirc x = [z]_{T\sncirc S}$. 

\end{toappendix}

 
A \emph{winning condition} of a game $A$ is 
a subset $W_A$ of  configurations $\iconf A$.
Informally, a strategy (for Player) is winning if it always prescribes moves for Player to end up in a winning configuration regardless of what opponent does. 
 Formally, $\sig:S\to A$ is a {\em winning strategy}  in
$(A, W_A)$, for a concurrent game $A$ with winning condition $W_A\subseteq \iconf A$,
 if  $\sig x$ is in $W_A$ for all   +-maximal configurations $x$ of $S$; a configuration is +-maximal if 
 the only additional moves enabled at it are those of Opponent.
 That $\sig$ is winning can be shown equivalent to all plays of $\sig$ against every counter-strategy of Opponent result in a win for Player~\cite{lics12,ecsym-notes}.

For the dual of a game with winning condition $(A, W_A)$, we again reverse the roles of Player and Opponent, 
and take its  winning condition to be the  
set-complement  of $W_A$, \ie~$(A,W_A)^\perp = (A^\perp,\iconf A\setdif W_A)$. 
In a 
parallel composition  
of two games with winning conditions, 
we deem a configuration 
winning if  its component in either game is winning: 
$(A, W_A) \vvbar (B, W_B) \eqdef (A\vvbar B, W)$,
where 
$W= \set{x\in\iconf{A\vvbar B}}{x_A\in W_A \hbox{ or } x_B \in W_B}$.  
With these extensions, we take a winning strategy from a game $(A,W_A)$ to a game $(B,W_B)$ 
to be a winning strategy in the game $A^\perp\vvbar B$, \ie~a strategy for which a win in $A$ implies a win in $B$.  
Whenever games are {race-free},  
copycat is a winning strategy; moreover, the composition of winning strategies is a winning strategy~\cite{lics12,ecsym-notes}.

\section{Games over relational structures}\label{sec:siggame}
We are concerned with games in which Player or Opponent moves instantiate variables (in a set $V$) with elements of a relational structure, subject to moves including challenges (in a set $C$) by either player.  
By restricting the size of the set of variables $V$, we can bound the memory intrinsic to a game, in the manner of pebble games
~\cite{pebblegames}.  

\begin{toappendix}
\subsection{Multi-sorted signature}\label{app:multisort}
A {\em relational many-sorted signature} $\Sigma$ is built from a set, whose elements are called {\em sorts}, which specifies for each string of sorts  $\vec s = s_1 \dots s_k$  a set $\Sigma_{\vec s}$ of {\em relation symbols} of {\em arity} $\vec s$. 
A many-sorted $\Sigma$-structure $\A$ provides sets and relations as interpretations of the sorts and relation symbols. 
Each element $a\in \mod A$ of a many-sorted $\Sigma$-structure $\A$ has a unique sort, $sort(a)$; we write $\mod\A_s$ for those elements of $\A$ of sort $s$; we insist on {\em nonemptiness}, \ie~$\mod\A_s\neq \emptyset$ for all sorts $s$. 
 A relation symbol $R\in \Sigma_{\vec s}$, of arity  $\vec s= s_1 \dots s_k$, is interpreted in $\A$ by $R_\A \subseteq \mod\A_{s_1} \times \dots \times  \mod\A_{s_k}$. 
The {\em sum} of many-sorted relational structures $\A_1$, with signature $\Sigma_1$, and  
$\A_2$, with signature $\Sigma_2$, denoted $\A_1+\A_2$ with signature $\Sigma_1+\Sigma_2$,  has sorts given by the disjoint union of the sorts of $\A_1$ and $\A_2$, and relations lifted from those of the components. 
\end{toappendix}

We will work with {\em many-sorted 
relational structures} because it will be useful to consider different relational structures, say $\A$ and $\B$, in parallel, via their sum $\A+\B$;  the sum has sorts given by the disjoint union of the sorts of $\A$ and $\B$, with relations and arities lifted from those of the components  [appendix \ref{app:multisort}].
 A \emph{game signature $(\Sigma, C, V)$} comprises $\Sigma$, a \emph{many-sorted} relational signature including equality for each sort;  
a set $C$  of {\em constants}; 
and a set $V=\setof{\al, \be,\ga, \dots}$ of  sorted {\em variables},
 disjoint from $C$,
with sorts in $\Sigma$ -- we write $sort(\al)$ for the sort of $\al\in V$. 
We assume each sort in a $\Sigma$-relational structure is associated with a nonempty set.

\begin{definition}[signature game]
{\rm 
A \emph{signature game} with signature $(\Sigma,C,V)$, alternatively a
\emph{$(\Sigma,C,V)$-game} $G$, comprises a race-free esp $G=(\mod G, \leq_G, \hash_G, \pol_G)$  and  
\begin{itemize}

\item a \emph{labelling} $\Var:\mod G\to V\cup C$ such that $g\co g'$ implies $\Var(g)\neq \Var(g')$; and $ \Var(g) = \Var(g')$ implies $\pol_G(g) = \pol_G(g')$;  
\item
a \emph{winning condition} $W_G$, which is an assertion in the free logic over $(\Sigma,C,V)$. 
(Free logic
is explained in detail in the Addendum, Section~\ref{add:freelogic}.)
\end{itemize}}
Events in $G$ are either {\em $V$-moves} or {\em $C$-moves}.  Note the elements of $C$ and $V$ appearing in  a \emph{$(\Sigma,C,V)$-game} $G$ are associated with a unique polarity. We shall specify the set of moves in $G$ labelled in $V_0\subseteq V$ by $\mod G_{V_0}$. 
We will consider $G$ to be 
played over a $\Sigma$-structure $\A$; the pair $(G,\A)$ is called \emph{a game over a $\Sigma$-structure $\A$}.
\end{definition}

%


\begin{example}[a multigraph homomorphism game]\label{ex:multigph}
The signature game in Figure~\ref{fig:bigraphs} is the multigraph homomomorphism game for two graphs $\A$ and $\B$ with red and green edge relations; so its signature has two sorts, one for $\A$ and one for $\B$, each with their equality, together with red and green edge relations.
The game comprises four 
events assigned variables, $\al_1, \al_2$ with sort that of $\A$ and $\be_1, \be_2$ with sort that of $\B$.  Its winning condition $W$ is written under the illustration of the game.
Looking ahead to Definition~\ref{def:siggamestrat},
a rigid deterministic winning strategy 
in this game over $\A+\B$ is necessarily  open and 
corresponds to a homomorphism  
$\A$ to $\B$. 

    \begin{figure}

\includegraphics[scale=0.30]{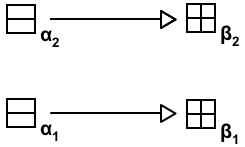}
\vspace{1em}
  \begin{equation*}
\begin{aligned}
  W\equiv \;&  (\al_1 =_\A \al_2 \to \be_1 =_\B \be_2) \wedge\\ 
  &(
  red_\A(\al_1, \al_2) \to red_\B(\be_1, \be_2)) \wedge \\
& (
green_\A(\al_1, \al_2) \to green_\B(\be_1, \be_2))
\end{aligned}
\end{equation*}
        \caption{The two-edge multigraph homomorphism game
        }
        \label{fig:bigraphs}
    \end{figure}
\end{example}
\begin{proposition}
Given a $(\Sigma,C,V)$-game, if for two events $g$, $g'$  in a configuration  $\Var(g) =\Var(g')$, then they are causally dependent, \ie~$g\leq_G g'$ or $g'\leq_G g$.
\end{proposition}

Consequently, the set of events in a configuration of $G$ which is labelled with the same variable or constant is totally ordered w.r.t.~causal dependency.  This fact is reflected in an esp $S$ via a total map $\sig:S\to G$.  For $x\in \iconf S$, it is sensible to 
define
\begin{align*}
    \last_S(x) \eqdef &\big\{s\in x\mid
\all s'\in x.\\& \big(s \leq_S s' \text{ and } \Var(\sig(s))=\Var(\sig(s'))\big) \implies s= s'\big\}.
\end{align*}
We shall write $\last_G$ when $\sigma$ is the identity on $G$.


Let $(G,\A)$ be a $(\Sigma,C,V)$-game over a $\Sigma$-structure $\A$.
A strategy in $G$ will assign values in $\A$ to $V$-moves, providing a dynamically changing environment for the variables---their values being the latest assigned. The dynamic environment, which  can both extend  and update the set of labelled events, is captured by the notion of an {\em instantiation}.

\begin{definition}[instantiation]{\rm
An {\em instantiation} of $G$ in $\A$ consists of $\sig:S\to G$, where $S$ is an esp and $\sigma$ is a total esp map; and a function $\rho:\mod S_V\to \mod A$,  
where $\mod S_V \eqdef \set{s\in \mod S}{\Var(\sig(s))\in V}$ and $sort(\rho(s)) = sort(\Var(\sig(s)))$.}
\end{definition}

A map from  instantiation $(\sig, \rho)$, where $\sig:S\to G$ and $\rho:\mod S_V\to \mod\A$,  to instantiation $(\sig',\rho')$, where $\sig':S'\to G$ and  $\rho':\mod {S'}_V\to \mod\A$, is a total esp map $f:S\to S'$ such that the following diagrams commute:
\[
\vcenter{\xymatrix{
S\ar[dr]_\sig \ar
[r]^{f}& S'\ar[d]^{\sig'}\\
&G}} \hspace{5em}
\vcenter{\xymatrix{
\mod S_V\ar[dr]_\rho \ar
[r]^{f}& \mod{S'}_V\ar[d]^{\rho'}\\
&\mod \A}}
\]

We will shortly construct a universal instantiation for a  $(\Sigma, C, V)$-game over a relational structure. 
Since total maps of event structures reflect causal dependency locally, the notion of last moves associated with variables or constants 
remains the same, whether according to the game or  the instantiation, and indeed across maps of instantiations.
\begin{proposition}\label{prop:last}
Let $S$ be an esp and let $\sig: S\to G$ be a total esp map.
Then, 
$\sig \last_S(x) = \last_G(\sig x)$.  
Moreover, if $f$ is a map of  instantiations from $(\sig:S\to G,\rho)$ to $(\sig':S'\to G,\rho')$, then  $f\last_S(x) = \last_{S'}(f x)$.  
\end{proposition}

An instantiation $(\sig, \rho)$, with $\sig:S\to G$ and $\rho:\mod S_V\to \mod A$, forms a model of the free logic.  
It specifies through $x\models_{\sig,\rho} \phi$ those configurations $x$ of $S$ which satisfy $\phi$. 
A novelty  is that a variable's value at $x$ is defined  by its latest move in $\last_S(x)$, if there is such. 
The naturalness of free logic stems from there being undefined terms. A variable may label an event in one configuration and not do so in another; without free logic, allowing undefined terms,  we would have to cope with undefinedness in a more ad-hoc way. Moreover, traditional predicate logic assumes a non-empty universe, which would make dealing with the empty configurations an issue.


Our semantics of the free logic ensures: 
\begin{propositionrep}\label{prop:truthInvariance}
Let $f$ be a map of instantiations of $G$ in $\mathcal A$, from $(\sig,\rho)$ to $(\sig',\rho')$.  Let $\phi$ be an assertion in free logic.  Then, 
$x\models_{\sig,\rho} \phi$ iff  $f x \models_{\sig',\rho'} \phi$, for every configuration $x$.

\end{propositionrep}
\begin{proof} 
By structural induction on $\phi$. Assume $\sig:S\to G$ and  $\sig':S'\to G$ with $f$ a map of instantiations from $(\sig,\rho)$ to $(\sig',\rho')$. The proof rests on $f\last_S(x) =\last_{S'}(f x)$, for any configuration of $S$, by Proposition~\ref{prop:last}.
\end{proof}


\section{Strategies 
over 
structures}\label{sec:strovrel}

A signature game over a $\Sigma$-structure $(G, \mathcal A)$ 
expands to a (traditional) concurrent game $\expn(G,\mathcal A)$ (Section~\ref{sec:games-strats}), from which we can derive strategies in games over relational structures. 
The idea is that each move associated with a variable $\al\in V$,  denoted $\boxempty_\al$, is expanded to all its  conflicting instances
$\xymatrix{\boxempty_\al^{a_1} \ar@{~}[r] &\boxempty_\al^{a_2} \ar@{~}[r]& \dots}$, where $a_i$ are elements of $\A$ with matching sort.
Precisely, paying more careful attention to causal dependencies---that such expansions can depend on earlier expansions, we arrive at the following definition. 
\begin{definition}[expansion]\label{def:expn}
 {\rm The {\em expansion} of a ($\Sigma$, C, V)-game over a $\Sigma$-structure $(G, \mathcal A)$ is the 
esp $\expn(G,\mathcal{A})$, with
\begin{itemize}
\item
events defined as pairs $(g,\gamma)$, where
$\gamma:[g]_V\to  \mod{\A}$  assigns an element of $\A$  of the correct sort to each $V$-move on which $g\in\mod G$ causally depends;
 \item causal dependency $\leq$ defined as
$(g',\gamma')\le(g,\gamma)$ iff $g'\le_G g$ and $\gamma'=\gamma\upharpoonright[g']_V$; 
\item \sloppy conflict relation $\hash$ defined as
$(g,\gamma)\hash (g',\gamma')$  iff  $g\,\hash_G\, g'$  or ${\exists g''\leq_G g, g'.\ \gamma(g'') \neq \gamma'(g'')}$; 
\item polarity map $\pol(g,\gamma) =\pol_G(g)$ inherited from $G$.
\end{itemize}}
\end{definition}

There is an obvious generalisation to an  expansion w.r.t.~$V_0\subseteq V$
; the role of variables $V$ above is simply replaced by the subset $V_0$. In Section~\ref{sec:charn}, we shall use a {\em partial expansion} w.r.t.~Opponent variables.

Let $\red: |\expn(G,\mathcal{A})| \to \mod G$ be the map such that $\red:(g,\gamma)\mapsto g$; it is an 
open map of event structures since each sort of $\A$ is nonempty. 
Let  $\Inst: |\expn(G,\mathcal{A})|_V \to \mod\A$ be the map such that $\Inst(g,\gamma)=\gamma(g)$ for $g\in\mod G$ with $\Var(g)\in V$.
%
The {\em  winning condition} of the expansion is  $$W\eqdef \set{x\in\iconf{\expn(G,\mathcal A)}}{x \models_{\red,\Inst} W_G}.$$
The expansion provides a universal instantiation of  $(G, \mathcal A)$.

\begin{lemmarep}[universal instantiation]\label{lem:universalinst}
The pair $(\red, \Inst)$ forms an {\em instantiation} of $G$ in $\A$.
For each instantiation $(\sig:S\to G, \rho)$ of $G$ in $\A$, there is a unique map $\sig_0$ of instantiations from $(\sig, \rho)$  to $(\red, \Inst)$ such that the following diagrams commute:
\[\vcenter{\xymatrix{
S\ar[dr]_\sig \ar@{-->}[r]^{\sig_0\ \ \ }& \expn(G, \A)\ar[d]^{\red}\\
&G}}\hspace{5em}
\vcenter{\xymatrix{
\mod S_V\ar[dr]_\rho \ar@{-->}[r]^{\sig_0\ \ \ }& \mod{ \expn(G, \A)}_V\ar[d]^{\rho_{G,\A}}\\
&\mod \A}}
\]

\noindent
 Moreover, $\red$ is open; $\sig$ is rigid iff $\sig_0$ is rigid; and $\sig$ is open iff $\sig_0$ is open.
\end{lemmarep}
\begin{proof}
For $s\in\mod S$, define $\sig_0(s) = (g,\ga)$ where $g=\sig(s)$ and for $g'\leq_G g$ with $\Var(g')\in V$,
$\ga(g') = \rho(s')$ for the unique $s'\leq_S s$  such that $\sig(s') = g'$.  Then, $\sig_0$ can be checked to be a map of event structures $\sig_0:S\to \expn(G,\A)$ and that it is the unique map of instantiations from $(\sig, \rho)$ to $(\red, \Inst)$. 
The map $\red$ is open because of the nonemptiness of $\A$.  
Because $\red$ is open, it is easy to check that $\sig_0$ is rigid/open according as $\sig$ is rigid/open, 
and {\it vice versa}.
\end{proof}

From the  universality of Lemma \ref{lem:universalinst} and Proposition~\ref{prop:truthInvariance}, the truth of an assertion w.r.t.~an instantiation reduces to its truth w.r.t.~the universal instantiation $(\red, \Inst)$ of the expansion.
We shall write
$x\models \phi$ for $x\models_{\red, \Inst}\phi$.
Given 
$x\in \iconf G$  and a sort-respecting $\rho:\mod x_V\to \mod\A$, write $x[\rho]\models \phi$ when $x\models_{j,\rho} \phi$ w.r.t.~the instantiation $(j,\rho)$ based on the 
inclusion map $j:x \hookrightarrow G$; we shall identify $x[\rho]$ with the configuration of $\expn(G,\A)$ obtained via universality as an image of $x$.

\begin{toappendix}
\subsection{Explicit characterisation of strategies on games over structures}\label{sec:inststrat}

\end{toappendix}

We can now give a central definition of this paper:

\begin{definition}[strategy in a game over a $\Sigma$-structure]\label{def:siggamestrat}
{\rm 
Let $(G, \mathcal A)$ be a game over a $\Sigma$-structure.
A {\em strategy in} $(G, \mathcal A)$ is an instantiation $(\sig, \rho)$ of $G$ in $\A$ for which $\sig_0$, the unique map of instantiations from Lemma~\ref{lem:universalinst}, is a concurrent strategy in $\expn(G, \mathcal A)$. 

Say $(\sig, \rho)$ is {\em rigid/open} according as $\sig_0$ is rigid/open. 
Say $(\sig, \rho)$ is {\em winning} if $\sig_0$ is winning in $\expn(G,\A)$ and 
{\em deterministic} if $\sig_0$ is deterministic. 
}\end{definition}

An explicit characterisation of strategies in $(G, \mathcal A)$ 
    can be given [appendix \ref{sec:inststrat}]. As expected, a  {\em strategy} in a game over $\A$ assigns  values in $\mathcal A$ to Player moves of the game $G$ in answer to assignments of Opponent.  It satisfies the same condition of innocence as concurrent strategies while the original condition of receptivity is generalised so that Opponent is free to make arbitrary assignments to accessible variables of negative polarity~\cite{sacha}.

Strategies in $(G, \mathcal A)$, a game over a $\Sigma$-structure, are related by maps of instantiations.
Lemma~\ref{lem:universalinst} provides an isomorphism between the ensuing category of strategies in $(G, \mathcal A)$  
and the category of strategies in $\expn(G,\A)$.

\setcounter{subsection}{1}
\subsubsection{The bicategories of signature games.}
Signature games are forms of game schema, useful for describing possible patterns of play without committing to particular instantiations.  Let $G$ be a $(\Sigma, C, V)$-game. Its \emph{dual} $G^\perp$  is the $(\Sigma, C, V)$-game 
obtained by reversing polarities and negating the winning condition. If $H$ is a $(\Sigma', C', V')$-game, the \emph{parallel composition} $G\vvbar H$ is the obvious $(\Sigma+\Sigma', C+C', V+V')$-game comprising  the parallel juxtaposition of event structures with winning condition the disjunction of those of $G$ and $H$ (strictly, after a renaming of constants and variables associated with the sums $C+C'$ and  $V+V'$).  
Consequently $G^\perp\vvbar H$ has winning condition the implication $W_G\to W_H$.  The bicategory of signature games $\GAMESIG$  
has as arrows from $G$ to $H$ concurrent strategies $\delta: D\to G^\perp\vvbar H$ ---strictly speaking in the esp of  $G^\perp\vvbar H$ ---with 2-cells and  composition inherited from that of concurrent strategies.  Notice
the esp $D$ inherits the form of a signature game from $G^\perp\vvbar H$ and itself describes a  pattern of strategical play, now from $G$ to $H$. The bicategory $\GAMESIG$  
describes schematic patterns of play between signature games.

 Let $(G, \mathcal A)$, 
 $(H, \mathcal B)$ be games over relational structures. 
 \sloppy 
 A \emph{winning strategy from} $(G, \mathcal A)$ \emph{to} $(H, \mathcal B)$, written 
 $(\sig,\rho):(G, {\mathcal A}) \profto (H, \mathcal B)$,
is a winning strategy in the game $(G^\perp\vvbar H, {\mathcal A} + {\mathcal B})$. 
Lemma~\ref{lem:universalinst} provides an isomorphism between strategies over structures and  concurrent strategies, so providing
a bicategory $\RED$ of signature games and strategies instantiated over relational structures [appendix \ref{app:RED}]. 
The identity, copycat strategy of $(G,\A)$, is the instantiation $(\cc_{G,A},\ga_{G,A})$ which corresponds to the copycat strategy of $\expn(G,\A)$. 
 Because winning strategies compose, a winning strategy $(\sig,\rho)$ in $\RED$ from $(G, {\A})$ to $(H, \B)$ is a form of \emph{reduction}. It reduces  the problem of finding a winning strategy in $(H, \B)$  to finding a winning strategy in $(G, {\A})$: a winning strategy in $(G, {\A})$ is a winning strategy from the empty game and structure $(\emptyset, \emptyset)$  to $(G, {\A})$;  its composition with $\sig$  is a winning strategy in $(H, \B)$.
\setcounter{subsection}{2}
\subsubsection{Examples of games.}
As remarked above, w.r.t.~a strategy in $\GAMESIG$ 
such as $\cc_G$, the copycat strategy on a signature game $G$, we obtain a signature game, in this case based on the esp $\CC_G$, describing a pattern of play of strategies from $G$ to $G$.   
All the following examples illustrate such copycat patterns of play. 
\begin{example}[multigraph game revisited]\label{ex:multigph-revisited} {\rm 
Example~\ref{ex:multigph} can be recast as a copycat signature game $\CC_G$ w.r.t.~a signature game $G$ with winning condition $W$.  This decomposes the original multigraph game to a game {\em from} $G$ {\em to} $G$,  so a form of function space.  The resulting game has winning condition $W_A\to W_B$, where in $W_A$ the variables of $W$ have been replaced by $\al_1, \al_2$ and in $W_B$ by $\be_1, \be_2$.
Here we are adopting a convention maintained throughout our examples. If $G$ has variables $V$ then $G^\perp\vvbar G$ and $\CC_G$ have variables the disjoint union $V+V$: we shall write $\al, \al_1, \al_2, \cdots$ for the variables in the left component, and $\be, \be_1, \be_2, \cdots$
for variables in that on the right. To ensure that the winning condition $(W_A\to W_B)$ implies the required homomorphism property, that each edge relation is preserved individually, we adjoin {\em challenge events} associated with each edge relation---a small adjustment needed for compositionality.

In detail, 
take $G$ to comprise challenge events, three constant-labelled Opponent moves  $e$, $r$ and $g$, pairwise in conflict with each other, in parallel with two variable-labelled concurrent Player moves $\plmove_{\beta_1}$ and $\plmove_{\beta_2}$.  Its winning condition is 
$$
\eqalign{
W_B\equiv\  
 &(\EX(e) \to \be_1 = \be_2) 
\ \wedge\cr
&(\EX(r) \to red(\be_1,\be_2))
\ \wedge\ 
(\EX(g) \to green(\be_1,\be_2)).
}
$$
The game $\CC_G$ is illustrated in Figure \ref{fig:multigraphgamerevisited}. 
An open winning deterministic strategy of $\CC_G$ over $\A+\B$, for multigraphs $\A$ and $\B$, has the form $(\sig,\rho)$ with $\sig:S\to\CC_G$.  
For any +-maximal configuration $x$ of $S$, we have   $x\models_{\sig, \rho}W_A\to W_B$.  
If $x$ contains a challenge $\opmove_e$, $\opmove_r$ or $\opmove_g$, by +-maximality it must also contain the corresponding $\plmove_e$, $\plmove_r$ or $\plmove_g$.  Hence if also $x\models_{\sig,\rho} W_A$, \ie~it establishes the equality or edge between $\al_1$ and $\al_2$  corresponding to the challenge, it will establish the corresponding equality or edge between $\be_1$ and $\be_2$.  
More completely, 
open winning deterministic strategies of $\CC_G$ over $\A+\B$
correspond to homomorphisms from $\A$ to $\B$. The key observation here is that because $(\sig,\rho)$ is deterministic and open, each assignment $\opmove_{\al_i}^a$ in any +-maximal configuration of $S$ always determines the same assignment $\plmove_{\be_i}^b$, and moreover, because equality is preserved, this assignment is the same whether $i=1$ or $i=2$; thus $(\sig,\rho)$ determines a function from $\mod\A$ to $\mod\B$ which, as argued above also preserves red and green edges.  
\begin{figure}
    \centering
\includegraphics[scale=0.30]{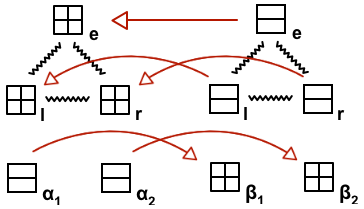}
    \caption{The multigraph game as a $\CC_G$ game}
\label{fig:multigraphgamerevisited}
\end{figure}
}\end{example}

\begin{example}[$k$-pebble game]\label{kpebble}
The (one-sided) $k$-pebble game \cite{kolaitisvardi} on two relational structures $\A$ and $\B$ is 
prominent in the field of finite model theory and 
the study of constraint satisfaction problems.
It can be explained as a copycat signature game $\CC_G$ w.r.t.~a signature game $G$.
Thus $\CC_G$ describes a pattern of strategy between  a copy $G^\perp$ associated with the structure $\A$ ---accordingly its variables are $\alpha_j$, for $j\in\{1,\dots,k\}$ ---and a copy $G$ associated with $\B$ ---its variables are  $\beta_j$, for $j\in\{1,\dots,k\}$ (Figure \ref{fig:kpebble}).  
The right copy $G$ describes sequences of potential assignments by Duplicator (Player) to variables $\beta_j$ as well as \emph{challenge events} of Spoiler (Opponent), labelled with a constant $c_i{(\vec \beta)}$, associated with realising a relation $R_i(\vec\be)$. Challenge events are all in conflict with each other and 
concurrent with all other events; again, the concurrency supported by 
event structures is essential. The winning condition of $G$ is 
$$W_B\equiv\bigwedge_{i, \vec\beta}  \EX(c_i{(\vec\beta)}) \to R_i(\vec \beta).$$
The conjunction above is over all pairs $i,\vec\beta$ such that  $R_i(\vec\beta)$ is well-sorted. 
Its existence predicate  $\EX(c_i{(\vec\beta)})$ is true only at configurations in which an event labelled $c_i{(\vec\beta)}$ has occurred.  So $W_B$ is true precisely at a configuration 
where any occurrence of a challenge $c_i{(\vec\beta)}$ is accompanied by latest instantiations which make $R_i(\vec\beta)$ true in $\B$.
The dual, left copy of $G^\perp$ has variables
 $\alpha_i$, awaiting an assignment in $\A$ by Spoiler (Opponent), and challenge events labelled $c_i{(\vec \al)}$ which are now moves of Duplicator (Player). Its winning condition is $\neg W_A$ where $W_A\equiv\bigwedge_{i,\vec\al}  \EX(c_i{(\vec\al)}) \to R_i(\vec \al)$ ---a conjunction over pairs $i,\vec\alpha$ with $R_i(\vec\alpha)$ well-sorted. 
The signature game $\CC_G$ comprises the  parallel composition $G^\perp\vvbar G$ with the additional 
  red arrows, further constraining the play so Duplicator awaits the corresponding move of Spoiler.  Its winning condition is $W_A\to W_B$. Its open deterministic winning strategies coincide with those of the traditional $k$-pebbling game.  
  The arbitrary and independent nature of Spoiler's challenges, events labelled $c_i{(\vec\al)}$, ensures  that in any +-maximal configuration,  if $R_i(\vec\al)$ in $\A$ then $R_i(\vec\be)$ in $\B$.
  Its rigid deterministic winning strategies are similar but need not force a Player assignment to $\vec\be$ in response to an assignment of Opponent to $\vec\al$ for which $R_i(\vec\al)$  fails.
\end{example}

\begin{figure}
   \centering
\includegraphics[scale=0.30]{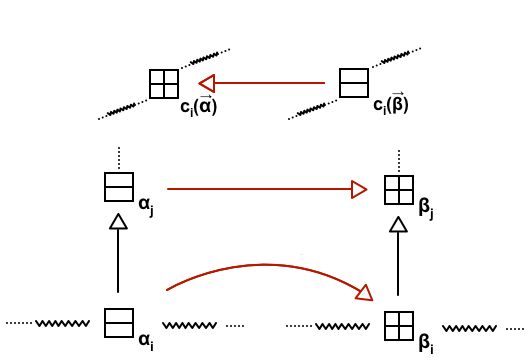}
    \caption{The $k$-pebble game}
    \label{fig:kpebble}
\end{figure}
\begin{example}[simulation game]\label{simulation}
The simulation game expresses the one-sided version of bisimulation \cite{park1981}. Its open 
winning deterministic strategies correspond to simulations. 
 We consider the simulation game (Figure \ref{fig:simulation}) on two  transition systems $(\mathcal A,a)$ and $(\mathcal B,b)$ with start states $a$ and $b$ respectively.
The game has two variables for each player.  A player makes assignments (in $\A$ for Spoiler and $\B$ for Duplicator) to infinite alternating sequences of their variables intertwined with challenges from the other player; the challenges $c_i$ specify transitions 
$R_i$.
Distinct sequences are in conflict with each other.
The red arrows represent the copycat strategy.
The two initial events  labelled by constants $\mathsf{st}$ are challenges associated with the $Start$ predicate used to identify the start states, which we assume have no transitions into them.
The simulation game can be constructed as a copycat signature game $\CC_G$, this time w.r.t.~a signature game $G$ with winning condition
$W_1$ below; the order constraints ensure that only the transitions in the direction of the transition system need be preserved: 
\begin{align*}
 (\EX(st) \to Start(\beta_1)) \ &\wedge
\bigwedge_{i} c_i\prel \be_2 \to R_i(\beta_1,\beta_2)
\\ &\wedge
\bigwedge_{i} c_i  \prel \be_1 \to R_i(\beta_2,\beta_1).  
\end{align*}
\end{example}
\begin{figure}
   \centering
\includegraphics[scale=0.30]{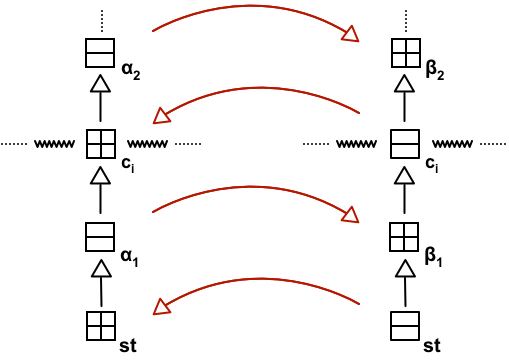}
    \caption{The simulation game}
    \label{fig:simulation}
\end{figure}
\begin{example}[Ehrenfeucht-Fra\"iss\'e game]{\rm 
As an example of a two-sided game, we present an Ehrenfeucht-Fra\"iss\'e game for checking the isomorphism of two $\Sigma$-structures. 
In such a game the aim of Player is to establish a partial bijection in which the relations of the structures are both preserved and reflected.  For this reason we extend the ``challenge'' events we have seen earlier in the k-pebble game with ``negative challenges'' associated with relations failing to hold; in addition to constants  $c_i(\vec\be)$, we have constants $nc_i(\vec\be)$ in a winning condition
$$
W_B\equiv \bigwedge_{
i, \vec\be} [(\EX(c_i(\vec\be)) \to R_i(\vec\be))\wedge
(\EX(nc_i(\vec\be)) \to \neg R_i(\vec\be)]
$$
---the conjunction above is over all pairs $i,\vec\beta$ s.t.~
$R_i(\vec\beta)$ is well-sorted. 
%

The Ehrenfeucht-Fra\"iss\'e game
is defined as a copycat signature game $\CC_{G}$, where the $(\Sigma, C, V)$-game $G = G_0\vvbar Ch$ comprises a subgame associated with challenge events $Ch$ in parallel with a subgame $G_0$.  
The subgame $Ch$ has events comprising all the constants $c_i(\vec\be)$ and $nc_i(\vec\be)$ in pairwise conflict with each other---they are all  Opponent moves. It is convenient to provide a recursive definition of $G_0$, borrowing from early ideas on defining event structures recursively in~\cite{icalp82}, where also the (nondeterministic) sum and prefix operation are introduced. From there, recall the relation $\tri$ between event structures: $E'\tri E$ iff $\mod{E'} \subseteq \mod E$ and the inclusion $E'\hookrightarrow E$ is a rigid map; then $\tri$ forms a (large) cpo w.r.t.~which operations such as sum, prefix and dual are continuous.  Using these techniques, we construct $G_0$ as the $\tri$-least solution to
$$
G_0= \opmove_l.\Sigma_{\be\in V}\plmove_\be.G_0 +
\opmove_r.\Sigma_{\be\in V}\opmove_\be.G_0^\perp.
$$
In $G_0$, initially Spoiler chooses to ``leave''---the move $\opmove_l$ ---or ``remain''---the move $\opmove_r$ ---at the current structure. If they remain, they assign a value in the current structure before the game resumes as the dual $G_0^\perp$; whereas if they leave, Duplicator assigns a value in the current structure and the game resumes as  $G_0$.

Through this definition of $G$, we achieve the following behaviour in $\CC_{G}$ which, recalling $\cc_G:\CC_G\to G^\perp\vvbar G$, relates assignments in $\Sigma$-structure $\A$ to variables $\al$ over $G^\perp$ on the left and assignments in $\Sigma$-structure $\B$ to corresponding variables $\be$  over $G$ on the right.  
(For simplicity, in Figure~\ref{fig:efgame} we only draw that part of $\CC_G$ involving $G_0$.)
Initially Spoiler chooses to leave $\opmove_l$ or remain $\opmove_r$ at the current structure $\B$. 
If Spoiler chooses to remain, they follow this by assigning a value in the structure $\B$  to a variable, a move $\opmove_\be$ over $G$; through the causal constraints of copycat, these moves are answered by Duplicator moves, the first of which acknowledges the choice of side while the second answers Spoiler's choice in $\B$ with a choice of element in $\A$, a move $\plmove_\al$. This is followed by a choice of Spoiler to leave or remain at $\A$ on the left. If alternatively Spoiler chooses to leave and play on the left, they will make an assignment in $\A$ and the next move on the right will be an assignment by Duplicator in $\B$. This accomplished, Spoiler resumes by again choosing a side on which to play.  
Open winning strategies of $\CC_G$ maintain partial isomorphisms between $\A$ and $\B$ whatever the moves of Spoiler. 
\begin{figure}
    \centering
    \includegraphics[scale=0.30]{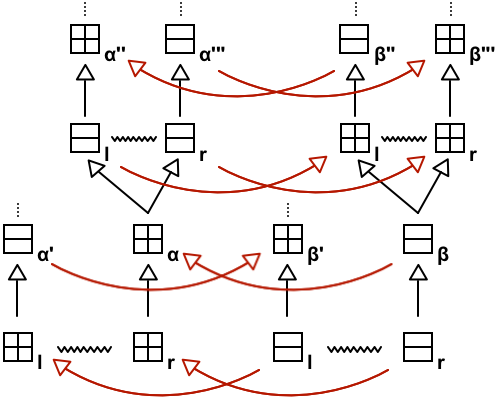}
    \caption{The Ehrenfeucht-Fra\"iss\'e game}
    \label{fig:efgame}
\end{figure}}
\end{example}

 \begin{toappendix}
Below, we use 
 $x\longcov s$ 
 to mean that an event $s$ is enabled at a configuration $x$, \ie~that $s\notin x$ and $x\cup\setof s$ is a configuration.
\begin{theorem}[explicit characterisation]\label{thm:stratcharn}
An  instantiation $(\sig:S\to G, \rho)$ is a strategy in $(G, \A)$  iff
  \begin{itemize}
\item 
for $x\in\iconf S$, $g\in\mod G$, with  $\sig x \longcov g$ and $\pol(g) =-$\,,
\begin{itemize}
\item
if $\Var(g)\in C$ then there is a unique $s\in \mod S$ such that $x\longcov s$ and $\sig(s) = g$;
\item
if $\Var(g)=\al\in V$ then for all $a\in \mod A$ with $sort(a) =sort(\al)$ there is a unique $s\in \mod S$ such that $x\longcov s$ and $\sig(s) = g$ and $\rho(s) =a$; 
\end{itemize}
\item
 if 
$s\imc_S s'$ and [$\pol(s) = +$ or $\pol(s') = -$], then 
$ \sigma(s)\imc_G \sigma(s')$.
 \end{itemize}
 \end{theorem}
 \begin{proof}
This follows directly from a strategy in a $(G,A)$ being the composition of a strategy in the expansion $\expn(G,\A)$ and $\red$, and properties of the latter~\cite{sacha}.
\end{proof}
\subsection{$\RED$, the bicategory of games over structures}\label{app:RED}

We explain how  composition in $\RED$  is inherited from composition in $\Strat$. 
It is easy to see that the expansion preserves the dual and parallel composition of games over relational structures.   Below, $(\emptyset, \emptyset)$ describes the empty game over the empty structure, with winning condition $\it false$, the unit w.r.t.~the parallel composition of games over structures.  

\begin{lemma}\label{lem:expnpres} Let $(G,\A)$ and $(H,\B)$ be games over relational structures.  Then, 
$$
\expn(\emptyset, \emptyset) = \emptyset\,, \ 
\expn((G,\mathcal{A})^\perp) = \expn(G,\mathcal{A})^\perp 
\hbox{ and }
\expn((G,\mathcal{A})\vvbar (H,\B)) \iso \expn(G,\mathcal{A})\vvbar \expn(H,\mathcal{B})\,.
$$
\end{lemma}

From Lemma~\ref{lem:expnpres} and Lemma~\ref{lem:universalinst},  which provides an isomorphism between strategies in games over relational structures and traditional concurrent strategies, we  
have  
$$
\RED((G,\mathcal{A}), (H,\B))
\iso
\Strat(\expn(G,\mathcal{A}), \expn(H,\B))\,.
$$
Write $(\sig,\rho): (G,\mathcal{A})\profto (H,\B)$ when $(\sig,\rho)\in \RED((G,\mathcal{A}), (H,\B))$.

We derive composition in $\RED$ from the composition of 
$\Strat$.  Let  $(\sig,\rho_S): (G,\mathcal{A})\profto (H,\B)$ and 
$(\tau,\rho_T): (H,\mathcal{B})\profto (K,\C)$.  By  Lemmas~\ref{lem:universalinst} and~\ref{lem:expnpres},   there are associated traditional strategies  $\sig_0: \expn(G,\mathcal{A})\profto \expn(H,\B)$ and 
$\tau_0: \expn(H,\mathcal{B})\profto \expn(K,\C)$.  Their composition $\tau_0\scirc \sig_0: \expn(G,\mathcal{A})\profto  \expn(K,\C)$ is a strategy in the game
$$\expn(G,\mathcal{A})^\perp\vvbar  \expn(K,\C)\iso \expn((G,\mathcal{A})^\perp\vvbar (K,\C))\,,$$ 
by Lemma~\ref{lem:expnpres}
; hence, by Lemma~\ref{lem:universalinst}, $\tau_0\scirc \sig_0$ determines a strategy $$(\tau,\rho_T)\scirc(\sig,\rho_S): (G,\mathcal{A})\profto  (K,\C)\,.$$

The composition $(\tau,\rho_T)\scirc(\sig,\rho_S)$ can be expressed directly,  
essentially mimicking the composition of traditional concurrent strategies, though with one key difference.  Instead of just forming the pullback of $\sig\vvbar K$ and $G\vvbar \tau$ we must restrict its configurations to those $y\sncirc x$ for which $\rho_S$ and $\rho_T$ agree over the common part $x_H=y_{H^\perp}$. Having done this, the construction of composition proceeds as for the composition of concurrent strategies, finally defining the resulting instantiation to agree 
with $\rho_S$   over the game $G^\perp$ and $\rho_T$ over the game $K$.
 \end{toappendix}

\section{Spoiler-Duplicator games 
}
 
We provide a general method for generating Spoiler-Duplicator games between relational structures; it goes far beyond the examples of the last section, all of which were based on copycat signature games.  
A category of Spoiler-Duplicator games $\SD_\delta$ is determined by a deterministic,  
idempotent comonad $\delta$ in the bicategory of signature games~\cite{Street}, 
so a deterministic strategy $\delta\in \GAMESIG(G,G)$ forming an idempotent comonad. The comonad $\delta$ describes a pattern of play: its counit ensures that this pattern respects copycat; its idempotence that the pattern is preserved under composition.
To characterise the strategies in $\SD_\delta$ we shall introduce the partial expansion of a signature game.

\setcounter{subsection}{1}
\subsubsection{Spoiler-Duplicator games deconstructed.}\label{sec:SDconstrn}
A  deterministic,  
idempotent comonad $\delta$ in $\GAMESIG$ provides us 
with a deterministic strategy $\delta:D\to G^\perp\vvbar G$, on a game $G$ with signature $(\Sigma, C, V)$, with counit $c:\delta\Rightarrow \cc_G$ and invertible comultiplication $d:\delta \Rightarrow \delta\scirc \delta$ satisfying the usual comonad laws.


 The  objects of $\SD_\delta$ are $\Sigma$-sorted 
relational structures $\A$, $\B$, etc. Its maps $\SD_\delta(\A,\B)$, written $(\sig,\rho): \A\profto_\delta\B$, are deterministic winning strategies $(\sig,\rho): (G, \A)\profto (G,\B)$, composing as in $\RED$, which factor {\em openly} 
through $\delta$.\footnote{Openness, rather than just rigidity, 
is essential to 
the existence of identities in $\SD_\delta$.}
To explain this in detail,  
note that $\sig$ is a map
$
\sig: S\to 
G^\perp\vvbar G$. 
By $\sig$ factoring 
openly through $\delta$, we mean that for some   
open map $\sig_1$ the following diagram commutes:
\[
\vcenter{\xymatrix{\
S\ar[dr]_{\sig} \ar@{-->}[r]^{\sig_1\ \ \ }& D\ar[d]^{\delta}\\
& G^\perp\vvbar G}}
\]
As $\delta$ is deterministic, it is mono~\cite{DBLP:journals/fac/Winskel12,ecsym-notes}, ensuring a bijective correspondence between strategies $(\sig,\rho): \A\profto_\delta \B$ and 
open
deterministic strategies $(\sig_1,\rho)$ in $(D,\A+\B)$; the correspondence respects the composition of strategies.
However, identities aren't in general copycat strategies and have to be constructed as pullbacks of $\delta$ along $\cc_{G,B}$, where 
$(\cc_{G,B},\ga_{G,B})$
is the copycat strategy  of $(G,\B)$\,:
\[
\vcenter{
\xymatrix{
I_B\pb{-45}\ar[r]^{\pi_2}\ar[d]_{\pi_1}& D\ar[d]^\delta\\
\CC_{G,B} \ar[r]^{\cc_{G,B}}& G^\perp\vvbar G }}
\]
Define $\iota_B\eqdef\delta\circ\pi_2=\cc_{G,B}\circ\pi_1$ and
$\rho_B\eqdef \ga_{G,B}\circ\pi_1$. 
Provided $(\iota_B, \rho_B)$ is winning, it will  constitute an identity in the category $\SD_\delta$ [see appendix \ref{composition}].
\begin{toappendix}
\subsection{Composition in $\SD_\delta$}\label{composition}
    Composition in $\SD_\delta$ is the composition of strategies over games over relational structures; that the composition factors openly through $\delta$ relies on horizontal composition preserving open 2-cells and the  idempotence of $\delta$.  
Identity strategies $(\iota_B,\rho_B):\B\profto_\delta \B$ must also factor openly through $\delta$.  They 
are constructed as pullbacks of $\delta$ along $\cc_{G,B}$, where 
$(\cc_{G,B},\ga_{G,B}): (G,\B)\profto (G,\B)$
is the copycat strategy  of $(G,\B)$\,:
\[
\vcenter{
\xymatrix{
I_B\pb{-45}\ar[r]^{\pi_2}\ar[d]_{\pi_1}& D\ar[d]^\delta\\
\CC_{G,B} \ar[r]^{\cc_{G,B}}& G^\perp\vvbar G \,.}}
\]
Define $\iota_B\eqdef\delta\circ\pi_2=\cc_{G,B}\circ\pi_1$ and 
$\rho_B\eqdef \ga_{G,B}\circ\pi_1$.  As the pullback of the strategy $\delta$
, the map $\pi_1$ defines a strategy in $\CC_{G,B}$; thus $\iota_B$, its composition with the strategy $\cc_{G,B}$, is a strategy.  That $I_B$ is deterministic---so the strategy $\iota_B$  is deterministic---follows from both 
$D$ and $\CC_{G,B}$ being deterministic.
To see that $(\iota_B,\rho_B)$ forms a strategy in $\SD_\delta$, fill in the diagram above to  the commuting diagram 
\[
\vcenter{\xymatrix{
I_B\pb{-45}\ar[r]^{\pi_2}\ar[d]_{\pi_1}& D\ar@(r,u)[dr]^\delta\ar[d]^c\\
\CC_{G,B} \ar@(d,d)[rr]^{\cc_{G,B}}
\ar[r]_{\ep_0}& \CC_G \ar[r]_{\cc_G}& G^\perp\vvbar G \,.
}}
\]
The map $\iota_B$ factors openly through $\delta$ as 
$\pi_2$ is open, being the pullback of the 
open map $\ep_0$.  

Provided $(\iota_B, \rho_B)$ is winning, the pair will together constitute an identity in the category $\SD_\delta$.  In general 
$\SD_\delta$ need not be a {\em subcategory} of $\RED$ as their identities may differ.

\begin{definition}[Scott order]\label{def:scottorder}  Let $A$ be a game and let  $x,y\in\iconf A$. 
We define the \emph{Scott order} $\bel_A$ as  $y\bel_A x\iff[\exists z\in \iconf A.\  y\supseteq^- z \text{ and } z\subseteq^+ x]$.
\end{definition}
It is not hard to show that $z$ is unique and equal to $x\cap y$. 
The Scott order is also characterised by $ y\bel_A x \iff  [y^-\supseteq x^- \text{ and }  y^+\subseteq x^+]$,
 which makes clear why it is a partial order.
 The Scott order is 
so named because it reduces to Scott's order on functions in special cases and plays a central role in relating games to Scott domains and ``generalised domain theory''~\cite{hyland-gendomthy}.
\begin{lemma}\label{lem:deltalem} Let $\delta:D\to G^\perp\vvbar G$ be the deterministic idempotent comonad $\delta:G\profto G$ in $\GAMESIG$ with comultiplication $d: D\iso D\scirc D
$.  Suppose $x, y, z\in \iconf D$.  If $d z= y\scirc x$, then   
  $$
z\uparrow_D x \ \& \  z\bel_D x \ \& \  
z\uparrow_D y \ \& \  z \bel_D y.$$  
\end{lemma}
\begin{proof} 
The comultiplication $d$ makes the diagram 
\[\xymatrix{
D\ar[r]^d \ar[d]_\delta&D\scirc D\ar[dl]^{\delta\scirc\delta}\\
G^\perp\vvbar G
}
\]
commute and 
is an isomorphism. Suppose $x, y, z\in \iconf D$ with $d z= y\scirc x$.  

We first show $z\uparrow_D x$.  As $\delta$ is deterministic it suffices to show no pair of  $-$ve events in $x\cup z$ are in conflict. 
It also suffices to consider a pair $e_1 \in x$ and $e_2 \in z$ with $\pol_D(e_1)=\pol_D(e_2)= -$. As $\delta$ is a strategy so receptive,
 if $e_1\hash_D e_2$ then $\delta(e_1) \hash_{G^\perp\vvbar G} \delta(e_2)$.  We show that $\delta(e_1) \hash_{G^\perp\vvbar G} \delta(e_2)$ is impossible.  Assume otherwise, that $\delta(e_1) \hash_{G^\perp\vvbar G} \delta(e_2)$.  Then $\delta(e_1)$ and $\delta(e_2)$ have to be in the same parallel component $G^\perp$ or $G$.  Consider the two cases.
 
\noindent
1. {\it Case where $\delta(e_1)\in \setof 1 \times G^\perp$ and  $\delta(e_2)\in \setof 1 \times G^\perp$.}  From the construction of $y\scirc x$ we see that $\delta(e_1), \delta(e_2) \in \delta x$, a configuration of $G^\perp\vvbar G$, which contradicts their being in conflict.  

\noindent
2. {\it Case where $\delta(e_1)\in \setof 2 \times G$ and  $\delta(e_2)\in \setof 2 \times G$.} Then from the construction of $y\sncirc x$ we see that $e_1\in x$ synchronises with a +ve event $\bar e_1\in y$. As the counit  factorises $\delta$ through $\cc_G$ there is $e_1'\in y$ with 
$\delta(e_1) = \delta(e_1') \in \delta z$.  But then $\delta(e_1), \delta(e_2)\in\delta z$, a configuration, contradicting their being in conflict. 

In all cases, $\delta(e_1)$ and $\delta(e_2)$ are not in conflict---a contradiction, from which we deduce that neither are $e_1, e_2$.  It follows that $x\cup z$ is a configuration and that $z\uparrow_D x$. 

Now, to show $z\bel_D x$ it suffices to show $\delta z\bel_{G^\perp\vvbar G} \delta x$.  Because then 
$$
\delta z \supseteq^- ((\delta z) \cap (\delta x)) \subseteq^+ \delta x,
$$
making  $\delta (z \cap  x)= (\delta z) \cap (\delta x)$, as $z\uparrow_D x$, so 
 $$
\delta z \supseteq^- \delta (z \cap  x) \subseteq^+ \delta x\,; 
$$
whence, as $\delta$ is deterministic, it reflects inclusions between configurations, to yield
$$
 z \supseteq^-  (z \cap  x) \subseteq^+ x, \quad \hbox{ \ie\  } \ z\bel_D x.
$$
To show $\delta z\bel_{G^\perp\vvbar G} \delta x$ we require
$$
\delta z^+ \subseteq  \delta x^+\  \ \&\  \ \delta z^- \supseteq  \delta x^-.
$$
Both these inclusions follow by considering over which component of $G^\perp\vvbar G$ events lie.  For example, to show $\delta z^+ \subseteq  \delta x^+$ the trickier case is where $e\in z^+$ lies over $G$, \ie~$\delta(e) \in \setof 2 \times G$.  Then, from the construction of $y\sncirc x$, we have $\delta(e)\in \delta y$.  But $\delta$ factors through $\cc_G$, from which there must be a matching $-$ve event $\bar e\in y$  which synchronises with a +ve event $e'\in x$ with $\delta(e) = \delta(e')$, as required. 
Similarly, $z\uparrow_D y$ and $z\bel_D y$.  
\end{proof} 
\begin{theorem} Suppose that the strategy $(\iota_A,\rho_A):(G,\A)\profto (G,\A)$  is  winning 
for all $\Sigma$-structures $\A$. Then
$\SD_\delta$ is a category with identities $(\iota_A,\rho_A):\A\profto_\delta \A$.  

Moreover, the category $\SD_\delta$ is a subcategory of $\RED$ in the case where $\delta$ is copycat $\cc_G:\CC_G\to G^\perp\vvbar G$;  then identities are 
copycat strategies $(\cc_{G,A},\ga_{G,A}): \A\profto_\delta \A$.
\end{theorem}
\begin{proof} 
Strategies $(\sig,\rho_S): \A\profto_\delta\B$ and $(\tau,\rho_T): \B\profto_\delta\C$ are associated with open strategies 
$\sig_1: S \to D$ and $\tau_1:T\to D$ and their composition with $\tau_1\scirc \sig_1: T\scirc S \to D\scirc D\iso D$, which is also open as horizontal composition of strategies $\scirc$ preserves open 2-cells.  

That $(\iota_\mathcal{A},\rho_\mathcal{A}):\A \profto_\delta \A$ and $(\iota_{\mathcal{B}},\rho_{\mathcal{B}}):\B \profto_\delta \B$ are identities in composition with $(\sig,\rho_S): \A\profto_\delta\B$ relies on
isomorphisms
$$
 I_{\mathcal{B}} \scirc S \iso S \ \hbox { and } \  S \scirc I_\mathcal{A} \iso S,
$$
where $\sig:S\to G^\perp\vvbar G$ and $\iota_{\mathcal{B}}: I_{\mathcal{B}}\to G^\perp\vvbar G$.  
We build the isomorphism 
$$
\theta: I_{\mathcal{B}} \scirc S \iso S  
$$
as the composite map
$$
\xymatrix{
I_B \scirc S \ar[r]^{\pi_1\scirc S\ \ \ \ }&  \CC_{G,B} \scirc S 
\iso S,
}
$$
which relies on $\cc_{G,B}$ being identity in $\RED$. (The map names are those used in the definition of $\iota_B$ above.) 
To show that $\theta$ is an isomorphism it suffices to show that it  
is rigid and both surjective and injective on configurations (Lemma~3.3 of\,\cite{ESS}). 
Injectivity on configurations is a consequence of   $\pi_1\scirc S$ being a deterministic strategy in 
$\CC_{G,B} \scirc S$---it is the composition of deterministic strategies.  
To establish rigidity and surjectivity on configurations, consider the commuting diagram
\[
\vcenter{\xymatrix{
I_B \scirc S\ar[d]_{
\pi_2\scirc\sig_1}  \ar[r]^{\pi_1\scirc S\ \ \ \ }&  \CC_{G,B} \scirc S \ar@{}[r]|{\iso} \ar[d]^{\cc_{G,B}\scirc \sig_1}&S \ar[d]^{\sig_1}\\
D \scirc D \ar[r]^{c \scirc D\ \ \ \ }&  \CC_{G} \scirc D \ar@{}[r]|{\iso} &D\ar@(d,d)[ll]_d&
}}\hspace{-2em},
\]
where $\theta$ is the upper horizontal composite map.  
The lower horizontal composite map is inverse to $d$ ---this follows from the comonad law for the counit $c$ ---and the map $c \scirc D$ 
an isomorphism.  

The vertical maps are open---being the 
composition of open 2-cells,  so rigid.   Consequently $\sig_1\circ \theta$ is rigid which, through $\sig_1$ reflecting causal dependency locally, implies the rigidity of $\theta$.  

To see that $\theta$ is surjective on configurations is more involved.  Suppose that $z_1\in\iconf S$.  Let $z\eqdef \sig_1 z_1 \in\iconf D$.  As $d$ is an isomorphism, there are $x, y\in \iconf D$ such that $d z= y\scirc x$.  By Lemma~\ref{lem:deltalem}, 
$$
z\bel_D x \ \hbox{ and } \   z \bel_D y.
$$
Because $\sig_1$ is deterministic and open there is a unique $x_1\in\iconf S$ with $z_1\bel_S x_1$ and $\sig_1 x_1 =x$.
We now construct $y_1\in\iconf{I_B}$ so that $\theta (y_1\scirc x_1) = z_1$.   
It is defined via an instantiation.

Let $y_0= c y$, a configuration of $\CC_G$; as $\cc_G:\CC_G \to G^\perp\vvbar G$ is given as the identity function on events, and $\delta = \cc_G
c$
we also have $y_0 \in G^\perp\vvbar G$.   
Define $\rho:y_0 \to \mod{\mathcal B}$ by cases, according as $e\in y_0$ is in the left or right component of $G^\perp\vvbar G$. 
If 
$e \in \setof 1 \times G^\perp$, then 
$\rho(e) \eqdef  \rho_S(e_1)$ when $e_1\in x_1 \ \&\  \sig(e_1)  = \bar e$. 
If
$e \in \setof 2 \times G$, then  
$\rho(e) = \rho_S(e_1)$ when $e_1\in z_1 \ \&\  \sig(e_1)  = e$. 
From this definition it can be checked  that $\rho(e) =\rho(\bar e)$ when $e, \bar e \in y_0$.  For this reason, $y_0$ and $\rho$ define a configuration $y'$ of $\CC_{G,\mathcal B}$.  The required $y_1$ is defined by $y_1 = y' \wedge y$.  By construction, $\theta (y_1\scirc x_1) = z_1$, as required for $\theta$ to be surjective on configurations. 

 We deduce that $\theta$ is an isomorphism  $I_{\mathcal{B}} \scirc S \iso S$.  Similarly,  $S \scirc I_\mathcal{A} \iso S$. This makes $(\iota_B, \rho_B)$ the identity w.r.t.~composition in $\SD_\delta$. 
\end{proof}
\end{toappendix}

By choosing a suitable $\delta$  we can design Spoiler-Duplicator games with different 
patterns of play and constraints on the way resources, \eg~number of variables, are used.  
For the $k$-pebble game (Example \ref{kpebble}) and the simulation game (Example \ref{simulation}), a simple choice of comonad $\delta$ being a copycat strategy $\cc_G$ 
suffices. 
Besides $\delta$ acting as copycat other behaviours of games can be found in the literature:
the all-in-one $k$-pebble game \cite{pebblerelationgame} is a version of the $k$-pebble game where Spoiler plays their entire sequence at once; then 
Duplicator answers with a similar sequence, preserving partial homomorphism in each prefix of the sequence.
Similarly, a game for trace inclusion \cite{glabbeek} is obtained from that for simulation by insisting Spoiler present an entire trace before Duplicator provides a matching trace.
In such games $\delta$  expresses that Duplicator awaits Spoiler's complete play before making their moves in a matching complete play.

 \begin{example}[Trace-inclusion game]\label{traceinclusion} 
 The trace-inclusion game in Figure~\ref{traceinclusiongame} is another example of an ``all-in-one'' game. It is an adaptation of the earlier simulation game, Example~\ref{simulation}, in two ways: so that Duplicator awaits all of Spoiler's choices of transitions 
 before making theirs; so that Duplicator answers correctly to every subsequence of Spoiler's choices, not only their latest pebble placements.  The former we ensure by using ``stopper'' events $\$$ to mark the explicit termination of Spoiler's assignments (primed); in any configuration the 
 start move
 on the right causally depends on a stopper $\$$ on the left.
For the latter, we make copies (unprimed) of the sequences of assignments. 
  \begin{figure}
        \centering
        \includegraphics[scale=0.30]{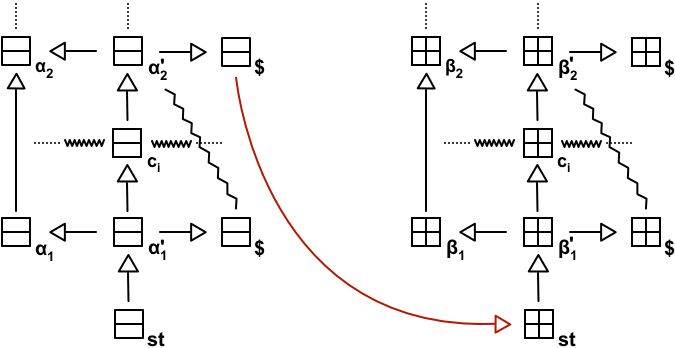}
        \caption{The trace-inclusion game}
        \label{traceinclusiongame}
    \end{figure}
In Figure \ref{traceinclusiongame} we have illustrated the comonad $\delta:D \to G^\perp\vvbar G$
in $\GAMESIG(G,G)$ for the trace inclusion game. The esp $D$ comprises the signature game $\CC_G$  with extra causal dependencies; for simplicity, we show only one causal dependency from a stopper and omit
copycat arrows. The winning condition of $G$ extends the winning condition $W_1$ of the simulation game:
\begin{align*}
   W_1 
&\wedge
\EX(\$) 
\wedge
(\EX(\beta_1) \to \beta_1' \pre \beta_1)
 \wedge
(\EX(\beta_2) \to \beta_2' \pre \beta_2).
\end{align*}
The 
two equality constraints ensure that the assignments to original sequences and their copies agree.  There is a winning strategy $(\A,a)\profto_\delta (\B,b)$ between two transition systems 
iff the traces of $(\A,a)$ are included in the traces of $(\B,b)$.  
A similar $\delta$ can be constructed for the  
pebble-relation comonad \cite{pebblerelationgame}---with $\SD_\delta$ isomorphic to its coKleisi category. 
   \end{example}

\setcounter{subsection}{2}
\subsubsection{Partial expansion.}\label{sec:partlexpn}

The  
expansion of Section~\ref{sec:strovrel} generalises
to an  expansion  w.r.t.~a subset of variables.  A key to understanding the strategies of $\SD_\delta$ is the {\em partial expansion} of $D$  w.r.t.~its Opponent moves [appendix \ref{app:partialexpn}].

\begin{toappendix}
\section{Partial expansion}\label{app:partialexpn}
\begin{definition}\label{def:V0expn}{\rm Let $G$ be a $(\Sigma, V, C)$-game.  Let $A$ be a $\Sigma$-algebra. 
Let $V_0\subseteq V$.  
The {\em $V_0$-expansion} of  $(G, \mathcal A)$  is the 
event structure $\expn^{V_0}(G,\mathcal{A})$ with 
\begin{itemize}
\item
{\em events} $(g,\gamma)$ where
$\gamma:[g]_{V_0}\to  \mod\mathcal{A}$  assigns an element of $\mathcal A$  of the correct sort to each ${V_0}$-move on which $g$ causally depends;
 \item
 {\em causal dependency} 
$(g',\gamma')\le(g,\gamma)$ iff $g'\le_G g\ \&\ \gamma'=\gamma\upharpoonright[g']_{V_0}$; 
\item
{\em conflict} 
$(g,\gamma)\hash (g',\gamma')$ \  iff  $g\,\hash_G\, g' \hbox{ or }\exists g''\leq_G g, g'.\ \gamma(g'') \neq \gamma'(g'')$; and
\item
{\em polarity} $\pol(g,\gamma) =\pol_G(g)$, inherited from $G$.
\end{itemize}
The function $\red^{V_0}: |\expn_{V_0}(G,\mathcal{A})| \to \mod G$ acts so $\red^{V_0}:(g,\gamma)\mapsto g$; it is an 
open map of event structures (because each sort of $\A$ is nonempty). 
The function  $\Inst^{V_0}: |\expn^{V_0}(G,\mathcal{A})|_{V_0} \to \mod\A$ acts so $\Inst^{V_0}(g,\gamma)=\gamma(g)$ for $g\in\mod G$ with $\Var(g)\in {V_0}$.
}\end{definition}

The pair $(\red^{V_0}, \Inst^{V_0})$ forms a universal  {\em ${V_0}$-instantiation} of $G$ in $\A$:

\begin{proposition}\label{prop:partlexpnuniversal}
Let $\sig: S\to G$ be a total map of event structures with polarity.  Let $\rho: \mod S_{V_0}\to \mod\A$ be a sort-respecting function from
$\mod S_{V_0}\eqdef \set{s\in \mod S}{\Var(\sig(s)) \in V_0 }. $
Then there is a unique map of instantiations $\sig_0:S\to \expn^{V_0}(G,\mathcal{A})_{V_0}$ from $(\sig, \rho)$ to $(\red^{V_0}, \Inst^{V_0})$.  
\end{proposition}
The above follows from the earlier Lemma~\ref{lem:universalinst} characterising expansion as giving the universal instantiation. 

\begin{definition}
The {\em partial expansion} $\expn^{-}(G,\mathcal{A})$ of $(G, \A)$ is $\expn^{V_0}(G,\mathcal{A})$ where $V_0\subseteq V$ is the subset of Opponent variables. 
\end{definition}

\begin{lemma}\label{lem:DApb}
The partial expansion $\DA$ is deterministic.  Moreover,  
if $y\cov^+ y'$ and $y\cov^+ y''$ in $\conf{\DA}$ and $\pi_1 y' = \pi_1 y''$, then, $y'=y''$.
\end{lemma}
\begin{proof}
From the construction of the pullback via secured families, and the fact that: if
$z\cov^+ z'$ and $z\cov^+ z''$ in $\conf{\expn(G^\perp, \A)}$ and $\red\, z' = \red\, z''$, then $z'=z''$.
\end{proof}

\end{toappendix}

Recall that $\delta$ is a map of esps $\delta:D\to G^\perp\vvbar G$ where $G$ has signature $(\Sigma, C,V)$ and winning condition $W$.
The signature of $G^\perp\vvbar G$, inherited by $D$, is $(\Sigma + \Sigma, V+V, C+C)$. 
In particular, the variables are either $V_1\eqdef \setof 1\times V$ of the left component, $G^\perp$, or $V_2 \eqdef  \setof 2\times V$ of the right component, $G$. 
Accordingly the moves of $D$ associated with variables are either  moves in the set $\mod D_{V_1}$ or  moves in $\mod D_{V_2}$.  
We form
the {\em partial expansion} of $(D, \A+\B)$ w.r.t.~$\Sigma$-structures $\A$ and  $\B$;
it expands  Opponent moves in $\mod D_{V_1}$ to 
instances in $\A$ 
and Opponent moves in $\mod D_{V_2}$ to 
instances in $\B$. 
 
The partial expansion of $(D, \A+\B)$, written $D(\A,\B)$,  is characterised as a pullback:
$$
\xymatrix{
D(\A, \B)\pb{-45}\ar[r]^{\quad\pi_1} \ar[d]_{\pi_2} &D\ar[d]^{\delta}\\
G^\perp(\A)\vvbar G(\B)
\ar[r]_{\quad\epsilon 
}&G^\perp\vvbar G
}
$$
Above, $G^\perp(\A)$ is the partial expansion of $(G^\perp, \A)$, while 
$G(\B)$  the partial expansion of $(G,\B)$, and $\epsilon$ the map projecting expansions to the original moves in $G^\perp\vvbar G$. 
Write $\rho(\A)$ for the sort-respecting function
$
\rho(\A): \mod\DAB_{V^-_1} \to \mod\A
$
sending $s\in\mod\DAB_{V^-_1}$ to $\rho_{G^\perp, \A}(e)$ when
$\pi_2(s) = (1, e)$; similarly,  $\rho(\B)$ is the sort-respecting function
$
\rho(\B): \mod\DAB_{V^-_2} \to \mod\B
$
sending $s\in\mod\DAB_{V^-_2}$ to $\rho_{G, \B}(e)$ when
$\pi_2(s) = (2, e)$. 
Write $\sig(\A,\B)$ for the composite map $\delta\circ\pi_1=\epsilon\circ\pi_2$.  


 \begin{toappendix}

 \begin{lemma}
The following diagram is a pullback:
$$
\xymatrix{
D(\A, \B)\pb{-45}\ar[r]^{\pi_1} \ar[d]_{\pi_2} &D\ar[d]^{\delta}\\
\expn^-(G^\perp\vvbar G,\A+\B)
\ar[r]_{\ \quad \qquad\qquad 
\epsilon^-_{G^\perp\vvbar G,\A+\B}
}&G^\perp\vvbar G
}
$$
where $\pi_1 \eqdef \epsilon^-_{D,\A+ \B}$ and $\pi_2$ is the unique map of instantiations from 
$(\delta\pi_1, \rho^-_{D,\A+\B})$ to $(\epsilon^-_{G^\perp\vvbar G,\A+\B}, \rho^-_{G^\perp\vvbar G,\A+\B})$.
\end{lemma}
\begin{proof}
To show it is a pullback consider  the following diagram (initially without $h$)
\begin{figure}[H]
    \centering
    \includegraphics[scale=0.25]{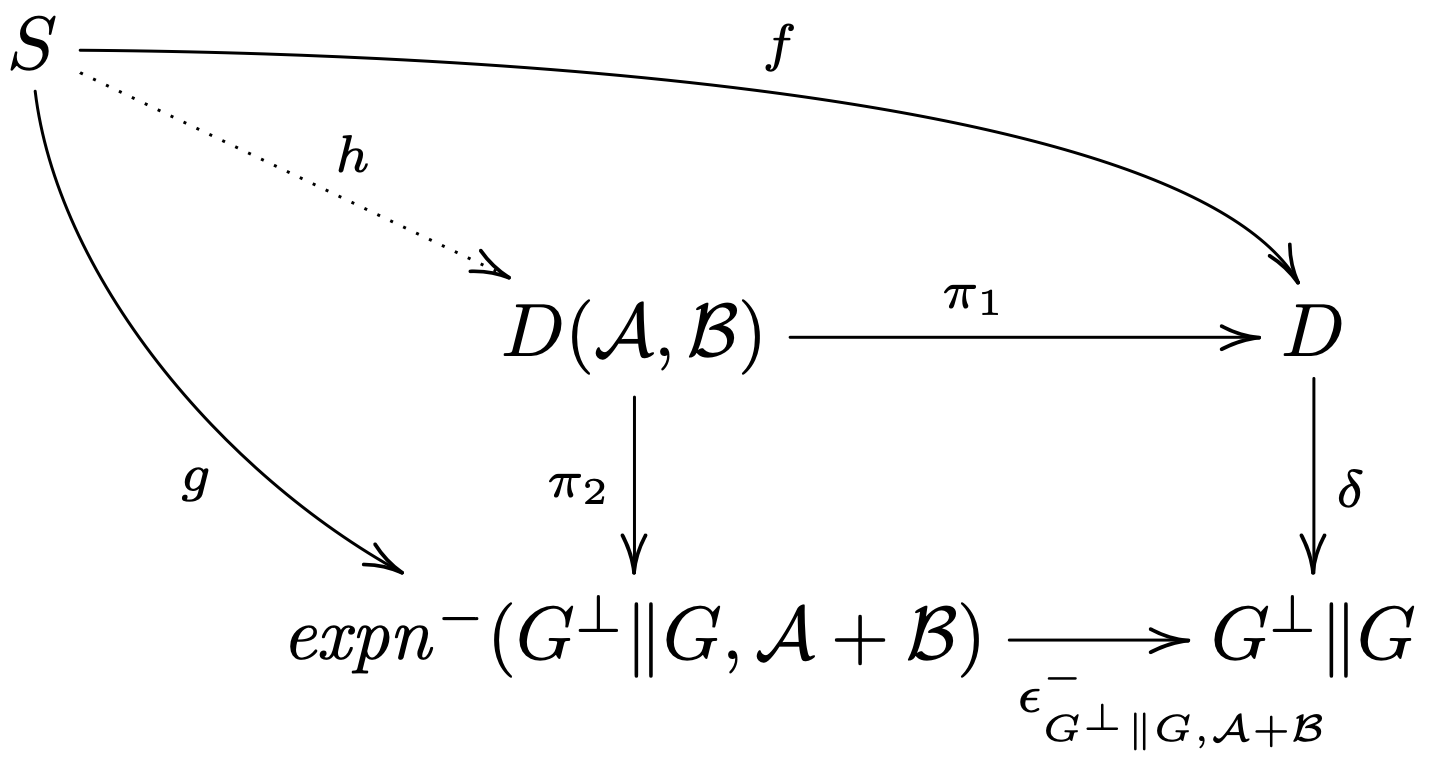}
\end{figure}
\noindent
where $f$ and $g$ are maps such that $\delta  f =g\,\epsilon^-_{G^\perp\vvbar G,\A+\B}$.  
Equip $S$ with $\rho \eqdef g\, \rho_{G^\perp\vvbar G, \A+\B}$ to form an instantiation $(f, \rho)$ in $D$.  From universality of the Opponent-expansion of $D$, there is a unique map of instantiations $h:S\to \DAB$ from $(f, \rho)$ to $(\pi_1, \rho^-_{D, \A+\B})$.  This ensures $f=\pi_1 h$.
Hence $\delta  f = \epsilon^-_{G^\perp\vvbar G,\A+\B}\,\pi_2 h$.  Both $g$ and $\pi_2 h$ are maps of instantiations from $(\delta f, \rho)$ to $(\epsilon^-_{G^\perp\vvbar G,\A+\B}, \rho_{G^\perp\vvbar G, \A+\B})$. From the latter's universality, the two maps must be equal, $g=\pi_2 h$.  This provides a mediating map $h:S\to \DAB$, as shown in the diagram above.  We also require its uniqueness.  Suppose $h':S\to \DAB$ is another mediating map, \ie~so $f= \pi_1 h'$ and $g=\pi_2 h'$.  Then those two conditions together ensure that $h'$ is a map of instantiations from $(f,\rho)$ to $(\pi_1,  \rho^-_{D, \A+\B})$.  From the latter's universality $h'=h$.  
\end{proof}

As 
$\pi_1 = \epsilon^-_{D,\A+ \B}$,  
it is open. As $\delta$ is a strategy, its pullback $\pi_2$ along $\epsilon^-_{G^\perp,\A}\vvbar \epsilon^-_{G,\B}$  a strategy in  $\expn^-(G^\perp,\A)\vvbar \expn^-(G,\B)$.  
Write $\sig(\A,\B)$ for the composite map 
$$
\sig(\A,\B)=\delta\circ\pi_1 = (\epsilon^-_{G^\perp,\A}\vvbar \epsilon^-_{G,\B})\circ\pi_2 : D(\A,\B)\to G^\perp\vvbar G\,.
$$ 
Write $\rho(\A)$ for the sort-respecting function
$
\rho(\A): \mod\DAB_{V^-_1} \to \mod\A
$
sending $s\in\mod\DAB_{V^-_1}$ to $\rho_{G^\perp, \A}(e)$ when
$\pi_2(s) = (1, e)$; similarly,  $\rho(\B)$ is the sort-respecting function
$
\rho(\B): \mod\DAB_{V^-_2} \to \mod\B
$
sending $s\in\mod\DAB_{V^-_2}$ to $\rho_{G, \B}(e)$ when
$\pi_2(s) = (2, e)$.  Then  $\rho_{D, \A+\B} = \rho(\A)\cup\rho(\B)$ as it acts as $\rho(\A)$ on $V_1^-$-moves  and as $\rho(\B)$  on $V_2^-$-moves.

 \end{toappendix}

Through the partial expansion $D(\A,\B)$ we can show the sense in which every strategy from $\A$ to $\B$  in $\SD_\delta$ follows essentially the same behavioural pattern: as we shall see,
 such a strategy corresponds to a function assigning values in $\A+\B$ to $(V+V)$-moves of Player in the partial expansion $D(\A,\B)$. 
Let  $(\sig,\rho): \A\profto_\delta \B$ be a strategy in $\SD_\delta$. Then $\sig$ factors through 
the 
expansion of $G^\perp\vvbar G$ w.r.t.~$\A$ and $\B$, as 
$$
\sig: S\arr{\sig_0} \expn(G^\perp, \A)\vvbar \expn(G, \B) \to G^\perp\vvbar G.
$$
Moreover, the expansion 
w.r.t.~$\A$ and $\B$ factors into the Opponent-expansion followed by the expansion w.r.t.~Player moves:
$$
\expn(G^\perp, \A)\vvbar \expn(G, \B) \to G^\perp(\A)\vvbar  G(\B) \to G^\perp\vvbar G.
$$
Consequently, we obtain the commuting diagram
$$
\xymatrix{
S\ar[d]_{\sig_0}\ar[rr]\ar@{..>}[rd]^\theta&&D\ar[d]^\delta\\
\expn(G^\perp, \A)\vvbar \expn(G, \B)\ar[rd]&D(\A,\B)\pb{-45}\ar[d]_{\pi_2} \ar[ru]^{\pi_1}&G^\perp\vvbar G\\
&G^\perp(\A)\vvbar  G(\B)\ar[ru]_\epsilon&
}
$$
where the dotted arrow represents the unique mediating map  $\theta$ to the pullback.   

\begin{lemmarep}\label{lem:two-wayiso}
The map  $\theta$ in the diagram above  is an isomorphism.
\end{lemmarep}
\begin{proof} 
It suffices to show that $\theta$ is rigid and both injective and surjective on configurations (Lemma~3.3 of \cite{ESS}).
As $(\sig,\rho)$ is a strategy in $\SD_\delta$, the map $\pi_1\theta: S\to D$ is open so rigid.  Consequently $\theta$ is rigid.  
The composite map 
$$
S\arr{\sig_0} \expn(G^\perp, \A)\vvbar \expn(G, \B) \to  \expn^-(G^\perp, \A)\vvbar \expn^-(G, \B)
$$
is a (nondeterministic) strategy---this relies on $\expn^-(G^\perp, \A)\vvbar \expn^-(G, \B)$ being an Opponent-expansion to ensure that further expansion, necessarily by values for Player moves, doesn't violate receptivity.  As remarked earlier, $\pi_2$ is also a strategy.  The map $\theta$, as a mediating 2-cell between these two strategies, is itself a strategy, so receptive as well as rigid.  
As $\sig$ is deterministic so is $S$, making $\theta$ a deterministic strategy and ensuring its injectivity on configurations.  
To show that $\theta$ is surjective on configurations  consider a covering chain
$$
\emptyset \cov y_1 \cov \cdots\cov y_i\cov y_{i+1} \cov  \cdots  \cov y_n=y
$$
to an arbitrary $y$ in $\conf{\DAB}$.  We show by induction along the chain that there are
$$
\emptyset \cov x_1  \cov \cdots\cov x_i\cov x_{i+1} \cov \cdots \cov x_n = x
$$
in $\,\conf S$ such that each $y_i =\theta x_i $.  Establishing this inductively for 
steps $y_i\cov^- y_{i+1}$ is a direct consequence of the receptivity of $\theta$.  
Consider 
a step $y_i\cov^+ y_{i+1}$ where inductively we assume $y_i= \theta x_i$. 
Then $\pi_1\theta\, x_i = \pi_1\,y_i \cov^+ \pi_1\, y_{i+1}$. From the openness of $\pi_1\theta$ there is $x_{i+1}$ with $x_i\cov^+ x_{i+1}$
such that
$\pi_1\theta \, x_{i+1} = \pi_1\, y_{i+1}$. 
Now,
$y_i\cov^+ \theta \, x_{i+1}$ and $y_i\cov y_{i+1}$ in $\conf{\DAB}$ and $\pi_1 \theta \, x_{i+1} = \pi_1 y_{i+1}$. 
By Lemma~\ref{lem:DApb}, $y_{i+1} = \theta \, x_{i+1}$, 
completing this inductive step.  \end{proof}

Via the isomorphism of Lemma~\ref{lem:two-wayiso}, we can reformulate the strategies of $\SD_\delta(\A,\B)$. 

\begin{theorem}\label{thm:two-sided}
Strategies $(\sig_0,\rho_0)\in\SD_\delta(\A, \B)$ 
are isomorphic 
to
strategies $$(\sig(\A,\B),
(\rho(\A) \cup k)\cup (\rho(\B) \cup h)
),
$$  where 
$$
k: \mod\DAB_{V_1^+} \to \mod\A
\hbox{ and }
 h: \mod\DAB_{V_2^+} \to \mod\B\,
 $$
 are sort-respecting functions
 from Player moves in $V_1$ and $V_2$ respectively, such that
 for all +-maximal $x\in\iconf\DAB$,
$$
x\models_{\sig(\A,\B), (\rho(\A) \cup k)\cup (\rho(\B) \cup h)
} W_1\to W_2.
$$
(Above, $W_1$ and $W_2$ are copies of the winning condition $W$ of $G$, renamed with  variables and constants of the left, respectively right, parts of the signature $(\Sigma+\Sigma, C+C, V+V)$.)
\end{theorem}
The theorem makes precise how  strategies of $\SD_\delta$ are determined by the extra part of the instantiation to variable-moves of Player  (given by $k$ and $h$ above).

Of course, we would wish to understand composition in terms of the above reformulation of strategies.  This is most easily done in the case of one-sided games, when game comonads and composition of strategies via coKleisli composition appear almost automatically out of the partial expansion. 
 
\section{
One-sided games}\label{sec:charn}

As above, let $G$ be a game with signature $(\Sigma, C, V)$ and winning condition $W$. 
Assume $\delta:D\to G^\perp\vvbar G$ is deterministic and forms an idempotent comonad $\delta\in\GAMESIG(G,G)$. 
The game $G$ with signature $(\Sigma, V, C)$ is called {\em one-sided} if its every $V$-move is a Player move. 
Then a strategy $(\sig,\rho)$ in $\SD_\delta$ from $\A$ to $\B$ provides assignments of variable moves in $\B$ by Player on the right  in response to assignments of variable moves in $\A$ by Opponent on the left; in this sense the strategies of $\SD_\delta$ are {\em one-sided} strategies. 
In most examples above, the games are one-sided. 
However, there are many well-known two-sided games, when the $V$-moves of $G$ are of mixed polarity, \eg~for bisimulation, trace equivalence and in Ehrenfeucht–Fraïssé games.
Throughout this section we focus on 
understanding $\SD_\delta$ where $\delta$ is a deterministic idempotent comonad and $G$ one-sided. 

The signature of $G^\perp\vvbar G$ is $(\Sigma + \Sigma, V+V, C+C)$. 
As earlier, the variables are either $V_1\eqdef \setof 1\times V$ of the left component, $G^\perp$, or $V_2 \eqdef  \setof 2\times V$ of the right component, $G$. 
Accordingly the moves of $D$ associated with variables are either Opponent moves in the set $\mod D_{V_1}$ or  Player moves in $\mod D_{V_2}$.  
Now
the partial expansion of $D$ is w.r.t.~just a single $\Sigma$-structure $\A$ and written $D(\A)$. A typical event of $D(\A)$ takes the form $(e,\ga)$ where $e\in \mod D$ and $\ga$ is a sort-respecting function from $[e]_{V_1}$ to $\mod\A$, reflecting that the expansion just instantiates $V_1$-moves. 
Specialising the earlier characterisation of the partial expansion as a pullback we observe:
\[
\vcenter{\xymatrix{
\DA\pb{-45}\ar[r]^{\pi_1} \ar[d]_{\pi_2} &D\ar[d]^{\delta}\\
\expn(G^\perp,\A)\vvbar G\ar[r]_{\qquad \red\vvbar G}&G^\perp\vvbar G
}}
\]
Write $\delta(\A)$ for the 
map 
$
\delta\circ\pi_1 
$. 
The function
 $\rho(\A)$ 
sends $(e,\ga)\in\mod\DA_{V_1}$ to $\ga(e)\in\mod\A$. 
 
 \setcounter{subsection}{1}
   \subsubsection{Obtaining game comonads from $\delta$.}\label{sec:comonads}
 We rely on the notion of {\em companion events} w.r.t a given map $\sig: S\to \CC_G$.  Recall that  ${\CC_G}$ comprises $ {G^\perp} \vvbar G$ with additional causal dependencies $\bar g \imc g$ or $g \imc \bar g$, across the components of the parallel composition, according to whether $g$ has positive or negative polarity.  
 Say $\bar s, s\in S$ are {\em companions} iff $\bar s, s\in S$ are not in conflict and $\sig(\bar s) = \overline{\sig(s)}$ in $\CC_G$. 
 Note that since $\sig$ reflects causal dependency locally, we have $\bar s \leq_S s$ or $s \leq_S \bar s$ according to whether $s$ is positive or negative. 

The counit of $\delta$ 
 is a map $c:D\to \CC_G$; so every $V_2$-move of Player over $G$  in $D$ causally depends on a companion over $G^\perp$ in $D$.  
 Composing the counit with the map $\pi_1: \DA \to D$, we obtain a map $c\circ \pi_1:\DA\to \CC_G$.
Whence,
every $V_2$-move $e$ of Player over $G$ in $\DA$ causally depends on a companion  $\bar e$ over $G^\perp$ in $\DA$.  
The event $e$ is not yet instantiated to a value in a $\Sigma$-structure $\B$ whereas, in this case, its companion $\bar e$ is instantiated to a value $\rho(\A)(\bar e)\in \mod\A$ by the Opponent expansion. 
We use their companions to equip the $V_2$-moves of Player over $G$ in $\DA$ with a relational structure.  

 
 Endow $\mod\DA_{V_2}$ with the structure of a $\Sigma$-structure $\Rel_\delta(\A)$ as follows:
 for a relation symbol $R$ of $\Sigma$ and $ e_1, \dots , e_k\in \mod\DA_{V_2}$ such that $R(\Var(e_1), \dots, \Var(e_k))$ is well-sorted,  we say $R(e_1, \dots , e_k)$  holds in 
$ \Rel_\delta(\A)$ iff
$$
\eqalign{
\exists x\in\iconf{\DA}.\  &\hbox{$x$ is +-maximal } 
\ \&\cr
 e_1, \dots ,  e_k\in\last(x) \ \&\ 
 &R_\A(\rho(\A)(\bar e_1), \dots, \rho(\A)(\bar e_k))
\ \&\ x_{G^\perp}\models W
}
$$ 
--- $x_{G^\perp}$ is the projection of $x$ to a configuration of $\iconf{\expn(G^\perp,\A)}$.

Generally, we can lift a sort-respecting function $h: \mod{\DA}_{V_2} \to \mod\B$ to a function $h^\dagger: \mod{\DA}_{V_2} \to \mod{\DB}_{V_2}$. 
To do so we cannot immediately use the universality of the Opponent expansion 
as this concerns the variables $V_1$.
Instead we have to rely on companion events in $D$ and the idempotence of $\delta$. 
Write $D_1$ for the set of events in $D$ over $G^\perp$, \ie~such that $\delta(e) = (1, g)$ for some $g\in G^\perp$.  
Similarly, write $D_2$ for those events in $D$ over $G$. 
Each Player event $e\in D_2$ has a unique companion $\bar e\in D_1$. 
The converse need not hold.   
However, from the idempotence of $\delta$ it follows that  
if $e_1\in D_1$, $e\in D_2$ and $e_1\leq_D e$, then there is a unique companion $\bar e_1\in D_2$ of $e_1$; moreover, $[\bar e_1]_{V_1} \subseteq [e]_{V_1}$ [appendix, Lemma~\ref{lem:companions}].

 Let  $h: \mod{\DA}_{V_2} \to \mod\B$ be a sort-respecting function to a $\Sigma$-structure $\B$. We use the concrete partial expansion w.r.t.~Opponent moves (see Definition~\ref{def:expn})  in defining the {\em  coextension} of  $h$ to a function $ h^\dagger: \mod\DA_{V_2} \to \mod\DB_{V_2}$. 
Let $(e,\gamma)$ be an event of $\mod\DA_{V_2}$; so $e\in D_2$ is a $V_2$-move with sort-respecting function $\gamma:[e]_{V_1} \to \mod\A$. 
We define $ h^\dagger
$  as
\[h^\dagger((e,\gamma)) \eqdef (e,\gamma')\,,\]  where $\gamma'(e_1) = h((\bar e_1, \ga\rstd [\bar e_1]_{V_1}))$
for all   $e_1\in [e]_{V_1}$.
  
\begin{definition}
     Say an assertion $\phi$ of the free logic  is {\em homomorphic} iff for all instantiations $(\sig, \rho)$ in $G$, with $\sig:S\to G$ and $\rho:\mod S_V \to \mod\A$ and homomorphisms $h:\A\to \B$, whenever 
$x\models_{\sig, \rho}\phi $ then $x\models_{\sig, h\circ\rho}\phi$. Say $G$ has {\em homomorphic winning condition} when its winning condition $W$ is homomorphic.
    \end{definition}
\begin{definition}
    Say $\delta$ is {\em balanced} when, w.r.t.~any $\Sigma$-structure $\A$, any +-maximal configuration $x$ of $D(\A)$ such that $x_{G^\perp}\models W$ has image $\delta(\A) x = x_0\vvbar x_0\in\iconf{G^\perp\vvbar G}$, for some configuration $x_0$ of $G$. (Note $\cc_G$ is balanced because, with $G$ being race-free, the +-maximal configurations of $\CC_G$ are precisely those of form $x_0\vvbar x_0$ with $x_0$ a configuration of $G$.)
\end{definition}

Under the assumptions that $G$ has homomorphic winning condition and $\delta$ is balanced,
if $h$ is a homomorphism $h:\Rel_\delta(\A) \to \B$, then its coextension $h^\dagger:\Rel_\delta(\A) \to \Rel_\delta(\B)$ is a homomorphism. In fact, $\Rel_\delta(\_)$ with counit $\Rel_\delta(\A) \to \A$ acting as $e\mapsto \rho(\A)(\bar e)$
and coextension $(\_)^\dagger$ forms a comonad on $\mathfrak R(\Sigma)$, the category of $\Sigma$-structures:
\begin{toappendix}

    \begin{lemma}\label{lem:companions}
\begin{enumerate}
 \item[(i)]
 \begin{enumerate}
 \item[(a)]
  If $e_1\in D_1$, $e_2\in D_2$ and $e_1\leq_D e_2$, then there is a unique companion $\bar e_1\in D_2$ of $e_1$\,;
   \item[(b)]
 if $e_1\in \DA_1$, $e_2\in \DA_2$ and $e_1\leq_\DA e_2$, then there is a unique companion $\bar e_1\in \DA_2$ of $e_1$\,;
  \end{enumerate} 
  \item[(ii)]
   if $e_0, e_1\in D_1$, $e_2\in D_2$ and $e_0\leq_D \bar e_1$ and $e_1\leq_D e_2$, then $e_0\leq_D e_2$\,.
  \end{enumerate} 
\end{lemma}
\begin{proof}
Follows from the idempotency of $\delta$. (W.l.o.g.~here we assume the isomorphism $d:\delta\Rightarrow \delta\scirc\delta$ is given by the identity function on events.)
\begin{enumerate}[(i)]
    \item (a) Otherwise the causal dependency of $D\scirc D$ would lack $e_1\leq e_2$ which needs the synchronisation of companions $e_1$ and $\bar e_1$,  and not  coincide with that of $D$.  (b) The analogous fact for $\DA$ follows from the openness of $\pi_1:\DA\to D$. 
The uniqueness of +ve companions of $-$ve moves, if they exist, follows from $D$ and $\DA$ being deterministic. 
 
    \item If $e_0, e_1\in D_1$, $e_2\in D_2$ and $e_0\leq_D \bar e_1$ and $e_1\leq_D e_2$, 
then $e_0\leq_{D\scirc D} e_2$ via the synchronisation of $e_1$ and $\bar e_1$ in the interaction, and hence $e_0\leq_D e_2$ by idempotency.  
\end{enumerate}
\end{proof}

\end{toappendix}
 \begin{theoremrep}[comonadic characterisation]\label{thm:partlexpncomonad}  
Assume $G$ has homomorphic winning condition and $\delta$ is balanced. The operation $\Rel_\delta(\_)$ extends to 
a unique comonad on $\mathfrak R(\Sigma)$, which
\begin{itemize}
\item 
maps every $\Sigma$-structure $\mathcal A$ to $\Rel_\delta(\A)$;
\item
has counit $\rho_\A: \Rel_\delta(\A)\to\A$ acting as $e\mapsto \rho(\A)(\bar e)$; 
\item
has coextension mapping a homomorphism $h: \Rel_\delta(\A)\to \B$ to a homomorphism $h^\dagger:\Rel_\delta(\A) \to \Rel_\delta(\B)$.\end{itemize}
\end{theoremrep}
\begin{proof}
 We require that a homomorphism $h: \Rel_\delta(\A)\to \B$ is taken to a homomorphism $h^\dagger:\Rel_\delta(\A) \to \Rel_\delta(\B)$  by the construction above.  Suppose the structural relation 
$R(e_1, \dots, e_k)$ in $\Rel_\delta(\A)$, \ie~ 
$$
e_1, \dots, e_k\in\last(x) \ \&\ R_\A(\rho(\A)(\bar e_1), \dots, \rho(\A)(\bar e_k))\ \&\ x_{G^\perp}\models W\,,
$$
for some +-maximal 
configuration $x$ of $\DA$.  We require that $R(h^\dagger(e_1), \dots, h^\dagger(e_k))$ in $\Rel_\delta(\B)$, \ie
$$
h^\dagger(e_1), \dots, h^\dagger(e_k) \in\last(y) \ \&\ R_\B(\rho(\B)(\overline{h^\dagger(e_1)}), \dots, \rho(\B)(\overline{h^\dagger(e_k)})) \ \&\ y_{G^\perp}\models W.
$$
for some +-maximal 
configuration $y$ of $\DB$.  (For brevity, we have written the elements of $\Rel_\delta(\A)$  as $e_1, \cdots, e_k$.  In the more detailed argument below, as events of the partial expansion, we shall describe them  more fully as taking the form $(e,\ga)$.)

Suppose $x$ is a +-maximal 
configuration of $\DA$ with $x_{G^\perp}\models W$.   An element of $x$ has the form $(e,\ga)$ with $e\in \pi_1^\A x$ and $\ga:[e]_{V_1} \to \mod A$.  We produce a configuration $y$ of $\DB$, essentially by updating each event $(e,\ga)$ to $(e,\ga')$, now with $\ga':[e]_{V_1} \to \mod B$, in accord with $h$.  In detail, update $(e,\ga)\in x$ to $(e,\ga')$ where 
$$
\all e_1\in [e]_{V_1}.\ 
\ga'(e_1) = h(\overline{(e_1, \ga\rstd[e_1]_{V_1})})
$$
---the updated value $\ga'(e_1)$
is the result of $h$ applied to the companion 
of $(e_1, \ga\rstd[e_1]_{V_1})$ in $x$; it is in knowing that this companion exists that
we use the assumption that $\delta$ is balanced.
When $(e,\ga)\in\mod x_{V_2}$, the companion is $(\bar e_1, \ga\rstd [\bar e_1]_{V_1})$, by Lemma~\ref{lem:companions}(ii), and
the update $(e,\ga')$ equals $h^\dagger((e,\ga))$.  
We define $y$ as the set of such updates of $(e,\ga)$ in $x$, \viz,
$$
y\eqdef \set{(e,\ga')}{
\exists \ga.\ (e, \ga)\in x \ \&\ 
\all e_1\in [e]_{V_1}.\ 
\ga'(e_1) = h(
\overline{(e_1, \ga\rstd[e_1]_{V_1})}
)}\,.
$$
Then, $y$ is a   +-maximal configuration of $\DB$, sharing a common image $\pi_1^\B y = \pi_1^\A x$ with $x$ in $D$.  The consistency and down-closure required of $y$ are inherited from $x$, as is $y$ being  +-maximal; via the open maps $\pi_1^\A$ and $\pi_1^\B$ which relate $x$ and $y$ through their common image in $D$.  

We now equip both $\mod x_{V_2}$ and $\mod y_{V_2}$ with $\Sigma$-structure.  
For any structural relation $R$, take $R((e_1, \ga_1), \cdots, (e_k, \ga_k))$ to hold in $\mod x_{V_2}$ iff
$$
(e_1, \ga_1), \cdots, (e_k, \ga_k) \in\last(x) \ \&\ R_\A(\ga_1(\bar e_1), \cdots, \ga_k(\bar e_k))\,. \eqno (A)
$$
Similarly, take $R((e_1, \ga'_1), \cdots, (e_k, \ga'_k))$ to hold in $\mod y_{V_2}$ iff
$$
(e_1, \ga'_1), \cdots, (e_k, \ga'_k) \in\last(y) \ \&\ R_\B(\ga'_1(\bar e_1), \cdots, \ga'_k(\bar e_k))\,.\eqno (B)
$$

We show that coextension restricts to a homomorphism between the $\Sigma$-structures,
$$
h^\dagger : \mod x_{V_2} \to \mod y_{V_2}\,.
$$
To see this, suppose $(A)$ above.  
We have that 
$h^\dagger((e_i, \ga_i))  = (e_i, \ga'_i)$, where $\ga'_i(e') = h((\bar e', \ga_i\rstd [\bar e']_{V_1}))$ for all $e'\in [e_i]_{V_1}$.
This  ensures 
$$ h^\dagger((e_i, \ga_i))\in\last(y) \hbox{ and } \ga'_i(\bar e_i) = 
h((e_i,\ga_i))\,.$$
Because $x_{G^\perp}\models W$ and $h:\Rel_\delta(\A) \to \B$ is a homomorphism, we get
$$
R_\B(h((e_1,\ga_1), \cdots (e_k, \ga_k)), \quad \ie~R_\B(\ga'_1(\bar e_1), \cdots, \ga'_k(\bar e_k))\,.
$$
Thus  
$(B)$
obtains and we have $R(h^\dagger(e_1, \ga_1), \cdots, h^\dagger(e_k, \ga_k))$  in $\mod y_{V_2}$.  Hence $h^\dagger : \mod x_{V_2} \to \mod y_{V_2}$ is a homomorphism between the $\Sigma$-structures,  as claimed.  To complete the argument that $h^\dagger$ is a homomorphism from $\Rel_\delta(\A)$ to $\Rel_\delta(\B)$, we also need that $y_{G^\perp} \models  W$ ---the purpose of the next stage of the argument.

As $\pi_1^\B y = \pi_1^\A x$ in $D$, they have a common image $z\eqdef \epsilon_{G^\perp,\B} y_{G^\perp} = \epsilon_{G^\perp,\A} x_{G^\perp}$ in $G^\perp$. 
Endow  $z$ with the causal dependence it inherits from $G^\perp$ to furnish an inclusion map
$i: z \hookrightarrow G^\perp$.  Any  $V_1$-move $g\in z$ is the image of a unique $(e,\ga)\in x$ so associated with its companion $V_2$-move $(\bar e, \ga') \in x$;  
define $\rho_x:  \mod{z}_{V_1} \to \mod x_{V_2}$ to take a $V_1$-move $g \in z$ to the companion $V_2$-move $(\bar e, \ga') \in x$. Together $(i, \rho_x)$ form an instantiation of $G^\perp$ in the structure $\mod x_{V_2}$. Composing with the homomorphism 
$h^\dagger : \mod x_{V_2} \to \mod y_{V_2}$ we obtain an instantiation of $G^\perp$ in the structure $\mod y_{V_2}$. 
Below we shall sketch the proof of the claim that 
$$
 x_{G^\perp} \models \phi   
 \iff z \models_{i,\rho_x} \phi \,, \eqno (C)
 $$
  in the free logic; and analogously,
 $$
 y_{G^\perp} \models \phi   
 \iff z \models_{i,\rho_y} \phi \,, 
 $$
 for $y$.  With this,  
 we argue
 $$
 \eqalign{
  x_{G^\perp} \models W &\iff  z \models_{i,\rho_x} W 
 \cr
 &
 \implies z \models_{i,h^\dagger\circ\rho_x} W\,, \hbox{ as $W$ is homomorphic,}
 \cr
 &
 \implies  z \models_{i,\rho_y} W\,, \hbox{ where } \rho_y = h^\dagger\circ\rho_x\,,
 \cr
 &
 \iff  y_{G^\perp} \models  W\,, 
 \cr}
 $$
where for the final step we reuse the claim, having observed that $\rho_y:  \mod{z}_{V_1} \to \mod y_{V_2}$ takes a $V_1$-move $g \in z$, the image of a unique $(e,\ga) \in y$, to the companion $V_2$-move $(\bar e, \ga') \in y$.

Because 
$x_{G^\perp} \models W$, we obtain $y_{G^\perp} \models W$.
Combined with the fact that $h^\dagger: \mod x_{V_2} \to \mod y_{V_2}$ is a homomorphism this demonstrates that, more globally, $h^\dagger$ is a homomorphism from $\Rel_\delta(\A)$ to $\Rel_\delta(\B)$.

It remains to tidy away the claim $(C)$, 
which follows by  structural induction on $\phi$, in a similar manner to Proposition~\ref{prop:truthInvariance}.  As an example, a key basis case is when $\phi$ is $R(t_1, \cdots, t_k)$.  Letting $(e,\ga)\in \mod x_{V_2}$ and $g$ be the image under $\delta$ of $\bar e$ in $G^\perp$, it  can be shown  that 
$$
(e, \ga)\in \last(x) \iff g \in \last(z)\,.
$$
Using this, it follows from their definition that 
$$
 x_{G^\perp} \models R(t_1, \cdots, t_k)   \iff z \models_{i,\rho_x} R(t_1, \cdots, t_k)\,.
 $$

To establish the comonad 
we need to verify the laws
$$
\Inst^\dagger = \id_{\Rel_\delta(\A)},
\quad
\Inst \circ h^\dagger=h, 
\quad 
(k\circ h^\dagger)^\dagger  = k^\dagger\circ h^\dagger.
$$
They follow routinely. 
\end{proof}


 We can simplify the definition of the comonad in the common case where $\delta$ is copycat; it is then based on the expansion of $G$.

\begin{corollaryrep}~\label{cor:SDcopycat} Assume $G$ has homomorphic winning condition $W$.
When $\delta=\cc_G$, under $e\mapsto \bar e$ the $\Sigma$-structure
$\Rel_\delta(\A)$ is isomorphic to  $\mod{\expn(G,\A)}_V$ with $\Sigma$-relations
$$
\eqalign{
&R(e_1, \cdots, e_k)\hbox{\   in \ } \mod{\exp(G,\A)}_V \hbox{\ iff \ }
\exists x\in\iconf{\exp(G,\A)}.\  \cr
& 
e_1, \cdots, e_k\in\last(x) \ \&\ R_\A(\Inst(e_1), \cdots, \Inst(e_k))\ \&\ x \models  W.
}
$$
Under this isomorphism 
the counit acts as 
$\Inst$
and 
coextension takes
$h: \mod{\expn(G,\A)}_V\to \mod\B$ to $h^\dagger: \mod{\expn(G,\A)}_V \to  \mod{\expn(G,\B)}_V$, where 
$h^\dagger(g,\ga) = (g,\ga')$ with $\ga'(g') = h(g', \ga\rstd [g']_G)$, 
for all $g'\in [g]_V$.  
\end{corollaryrep}
\begin{proof}  
For  $x\in\iconf{\exp(G,\A)}$, write $$R_x(e_1, \cdots, e_k) \ \hbox{ iff }\ 
e_1, \cdots, e_k\in\last(x) \ \&\ R_\A(\Inst(e_1), \cdots, \Inst(e_k))\ \&\ x \models  W.
$$

By definition, we have
$R(e_1, \cdots, e_k)$ in $\Rel_{\cc_G}(\A)$
iff
there exists +-maximal $y\in \iconf{\CC_G(\A)}$  with $e_1, \cdots, e_k\in \last(y)$ and 
$R_\A(\rho(\A)(\bar e_1), \cdots, \rho(\A)(\bar e_k))$ and $y_{G^\perp}\models W$. 
Taking $ x= y_{G^\perp}\in \iconf{\expn(G,\A)}$ 
we obtain $R_{x}(\bar e_1, \cdots, \bar e_k)$. 
Then, $R(\bar e_1, \cdots, \bar e_k)$ in $\mod{\exp(G,\A)}_V$.  

Conversely, if $R(\bar e_1, \cdots, \bar e_k)$ in $\mod{\exp(G,\A)}_V$ then 
$R_{\bar x} (\bar e_1, \cdots, \bar e_k)$ for some configuration $\bar x$ of $\expn(G,\A)$. 
Letting $x\eqdef \red \bar x$ we construct a +-maximal configuration $y
\in \iconf{\CC_G(\A)}$ as follows.
The map $\red\vvbar \id_x : \bar x\vvbar x \to G^\perp\vvbar G$ together with 
$\rho: (\bar x\vvbar x)_{V_1} \to \mod\A$, taking $(1,\bar e)$ to $\Inst(\bar e)$, form an instantiation of (the left part of) $G^\perp\vvbar G$ in $\A$.  Via the universality of partial expansion~\ref{prop:partlexpnuniversal} we obtain a map
$\bar x\vvbar x \to \CC_G$.  Its image provides $y \eqdef (\bar x\vvbar x)[\rho] \in \iconf{\CC_G(\A)}$. Then,
 $y_{G^\perp}=\bar x$ so $y_{G^\perp}\models W$ and---writing $e_i$ for the companion of $\bar e_i$ in $y$ ---we have $e_1, \cdots, e_k\in \last(y)$ with
$R_\A(\rho(\A)(\bar e_1), \cdots, \rho(\A)(\bar e_k))$. Hence $R(e_1, \cdots, e_k)$ in $\Rel_{\cc_G}(\A)$.

The re-expression of the counit and extension follow directly from the bijection 
$\mod{\Rel_\delta(\A)}\iso\mod{\expn(G,\A)}_V$ given by $e\mapsto \bar e$. 
\end{proof}
In particular, from 
Corollary~\ref{cor:SDcopycat} we obtain that the pebbling
 comonad $\mathbb{P}_k$ from~\cite{AbramskyDawarWang} arises as a special case:  $\mathbb{P}_k ({\mathcal A}) \iso  \Rel_{\cc_G}(\A)$,
where $k$ is the size of the set of variables in the $k$-pebble game $G$  of Example~\ref{kpebble}. 
Similarly, from Theorem \ref{thm:partlexpncomonad}  and its corollary, we can obtain other well-known game comonads such as the modal comonad \cite{relatingstructure}, pebble-relation comonad \cite{pebblerelationgame}, and the Ehrenfeucht–Fra\"iss\'e commonad \cite{relatingstructure}. 
This raises the question of when a coKleisli category is isomorphic to $\SD_\delta$, and what its relationship with $\SD_\delta$ is. More broadly, we seek a characterisation of one-sided strategies.  
In general, they are associated with a more refined preservation property than that of homomorphism, \viz~that winning conditions are preserved. 
This will be captured in Section \ref{sec:onesidedstrategies} by the concept of {\em $\delta$-homomorphism}. It specialises to usual homomorphisms w.r.t.~
appropriate winning conditions. 
 \setcounter{subsection}{2}
\subsubsection{Eilenberg-Moore coalgebras.}
Arboreal categories axiomatise many of those categories of Eilenberg-Moore coalgebras that arise from game comonads~\cite{arboreal}.  In this section we explore the Eilenberg-Moore coalgebras of $\Rel_\delta(\_)$ for a balanced deterministic idempotent comonad  $\delta:D\to G^\perp\vvbar G$ on a one-sided game $G$ with homomorphic winning condition; and see that they correspond to certain event algebras.  Categories of event algebras 
specialise to arboreal categories when the $V$-moves form a tree-like event structure and to linear arboreal categories \cite{lineararboreal} when they form a path-like event structure.

Assume $G$ has homomorphic winning condition and $\delta$ is balanced.
As might be expected, the coalgebras of $\Rel_\delta(\_)$ can be represented by $\Sigma$-structures which also take an event-structure shape, a characterisation we now make precise.   
 First note that 
 $\Rel_\delta(\A)$ inherits the form of an event structure, $D(\A)\downarrow\mod{\Rel_\delta(\A)}$ from the partial expansion $D(\A)$ w.r.t.~a $\Sigma$-structure $\A$;   
  the event structure, also denoted by $\Rel_\delta(\A)$, is the projection of $D(\A)$ to its +ve assignments, all in $V_2$.   The construction is functorial: given a $\Sigma$-homomorphism $h:\A\to\B$, the function $\Rel_\delta(h)$ is both a homomorphism and a rigid map of event structures.  

  The coalgebras of $\Rel_\delta(\_)$ correspond to $\delta$-event algebras, \viz~
 pairs $(\A, f)$ consisting of the $\Sigma$-structure $\A$ and a
rigid map of event structures, 
 $$f : (\mod\A, \leq_\A, \hash_\A) \to \Rel_\delta(\A),$$ 
whose domain is an event structure with events the elements of $\A$, such that whenever $f(a) =(e,\ga)$,
\[
\eqalign{
& \ga(\bar e) = a \ (\hbox{so $f$ is injective)} \hbox{ and}\cr
& e_1 \in [e]_{V_1} \implies  f(\ga(e_1))= (\bar e_1, \ga\rstd [\bar e_1]_{V_1}), \cr
}
\eqno (i)\]
and whenever $R_\A(a_1,\cdots, a_k)$, for $R$ in the signature,
\[
\eqalign{
&\exists x\in\iconf{D(\A)}.\  x\hbox{ is +-maximal}\ \&\ \cr
&
f(a_1),\cdots, f(a_k)\in \last_{D(\A)}(x) \ \&\  x_{G^\perp}\models W
.}
\eqno (ii)\]
 Suppose $$
f : (\mod\A, \leq_\A, \hash_\A) \to \Rel_\delta(\A)
\hbox{ and } g: (\mod\B, \leq_\B, \hash_\B) \to \Rel_\delta(\B)
$$
are $\delta$-event algebras.
A map  
from $f$ to $g$ is a map of event structures $h: (\mod\A, \leq_\A, \hash_\A) \to (\mod\B, \leq_\B, \hash_\B)$,
which is also a homomorphism $h:\A\to\B$ such that $\Rel_\delta(h) f= g h$. 
Such maps compose (and are necessarily rigid). 

\begin{toappendix}
    \subsubsection{Eilenberg-Moore coalgebras.}

In proving the characterisation of coalgebras of $\Rel_\delta$ it will be convenient to use the following  proposition~\cite{icalp82,evstrs}:

\begin{proposition}\label{prop:mapchar}
Let $A$ and $B$ be event structures.  A partial function $f:\mod A \parrow \mod B$ is a map of event structures iff
\begin{itemize}
\item
$f(a) = b \ \&\ b'\leq b
\implies 
\exists  a'\leq a.\ f(a') = b'$, for all $a\in\mod A, b' b'\in \mod B$; and
\item
$f(a) = f(a') \hbox{ or } f(a) \hash f(a') \implies a = a'   \hbox{ or } a\hash a'$, for all $a'a'\in\mod A$. 
\end{itemize}
\end{proposition}
\end{toappendix}

\begin{theoremrep}\label{thm:coalgs}
$\delta$-Event algebras and their maps form a category isomorphic to the category of Eilenberg-Moore coalgebras of $\Rel_\delta(\_)$.
\end{theoremrep}

\begin{proof}
Abbreviate $\Rel_\delta$ to $\R$. 
Let $f:\A\to \R(\A)$ be an Eilenberg-Moore coalgebra of the comonad $\R(\_)$. By definition, we have the commuting diagrams
$$
\xymatrix{
\A\ar[r]^f\ar[d]_f& \R(\A)\ar[d]^{\nu_\A}\\
\R(\A) \ar[r]_{\R(f)}& \R(\R(\A))
}
\qquad\hbox{ and }\qquad 
\xymatrix{
\A\ar[dr]_{\id_\A}\ar[r]^f& \R(\A)\ar[d]^{\rho_\A}\\
&\A\,.
}
$$
Comultiplication $\nu_\A$ is obtained as $(\id_{\R(\A)})^\dagger$ from which we derive
$$
\nu_\A(e,\ga) = (e, \ga')
 $$
 where, for all $e_1\in [e]_{V_1}$,
 $$
 \ga'(e_1) = (\bar e_1, \ga\rstd [\bar e_1])\,.
 $$

Let $a\in\mod A$ and 
suppose $f(a) =(e,\ga)$. From the second commuting diagram we get
$$
\ga(\bar e) =a\,. \eqno (1)
$$
It follows directly that $f$ is injective.  
From the first commuting diagram we obtain
$$
\R(f)\circ f (a) = \nu_\A\circ f(a)\,.
$$
But 
$$
\R(f)\circ f (a) = \R(f) (e,\ga) = (e, f\circ \ga)
$$
while 
$$
 \nu_\A\circ f(a) = \nu_\A(e,\ga) = (e, \ga')
 $$
 where, for all $e_1\in [e]_{V_1}$,
 $$
 \ga'(e_1) = (\bar e_1, \ga\rstd [\bar e_1])\,. 
 $$
 Thus, from the commuting first diagram we obtain
 $$
\all e_1 \in [e]_{V_1}\ f(\ga(e_1)) = (\bar e_1, \ga\rstd [\bar e_1])\,. \eqno (2)
$$

Construct an event structure $(\mod\A, \leq_\A, \hash_\A)$ on the elements of $\A$.
 For $a, a'\in\mod\A$, take 
 $$
a'\leq_\A a \hbox{ iff } f(a') \leq f(a)\  \hbox{ and } \ 
a'\hash_\A a \hbox{ iff }  f(a') \hash f(a)\,, \hbox{ in } D(\A)
\,.
$$
Now, from a special case of (2), if $f(a) = (e, \ga)$ and $(e',\ga') \leq (e,\ga)$ in $D(\A)_{V_2}$ then $\bar e'\in [e]_{V_1}$, so
 taking $a'=\ga(\bar e')$ we obtain $a'\leq_\A a$ and $f(a') = (e', \ga')$. As $f$ is also injective and  reflects conflict it is a map of event structures, by Proposition~\ref{prop:mapchar}.
The function $f$ is an injective and a rigid map of event structures, \ie~a rigid embedding 
$$f: (\mod\A, \leq_\A, \hash_\A) \to \R(\A)\,,$$ 
 which additionally satisfies (1) and (2).
  

Moreover, suppose  $R_\A(a_1,\cdots, a_k)$.   Then
\[
\exists x\in\iconf{D(\A)}.\  x\hbox{ is +-maximal}\ \&\ 
f(a_1),\cdots, f(a_k)\in \last_{D(\A)}(x) \ \&\  x_{G^\perp}\models W\,
\eqno(3)
\]
---this follows from $f$ being a homomorphism and the interpretation of $R$ in $\R(\A)$.

Conversely, a rigid map which satisfies (1) and (2) checks out to be a coalgebra.

The category of Eilenberg-Moore coalgebras of $\R(\_)$ has maps homomorphisms $h$ making
$$
\xymatrix{
\A\ar[r]^f \ar[d]_h&\R(\A)\ar[d]^{\R(h)}\\
\B\ar[r]_g&\R(\B)}
$$
commute, a fact which is left undisturbed when lifted to the required maps of event structures.
\end{proof}

Coalgebras have a simpler characterisation when $\delta$ is copycat:
\begin{corollary}\label{co:coalgebras}
Let $G$ be a signature game with homomorphic winning condition. Let $G_V$ be its projection to $V$-moves. A
$\cc_G$-event algebra corresponds to a
pair $(\A, f_0)$ consisting of a $\Sigma$-structure $\A$ and a
rigid map of event structures 
 $$f_0 : (\mod\A, \leq_\A, \hash_\A) \to G_V,$$ 
whose domain is an event structure with events 
elements of $\A$, such that whenever $R_\A(a_1, \cdots, a_k)$,   for a $\Sigma$-relation $R$,  
  \[
 \eqalign{
\exists x\in \mathcal C&(\expn(G,\A)).\  x\models W \ \& \ \cr &f(a_1), \cdots, f(a_k) \in \last_{\expn(G,\A)}(x).\cr}
\eqno (ii)'
\]
Above, $f:\mod\A\to \mod{\expn(G,\A)}$ is defined by  $f(a) = (f_0(a), \ga)$, where $\ga(e_1)$ is that unique $a_1\leq_\A a$ for which $f_0(a_1) = e_1$. 
\end{corollary}

A benefit of game comonads is their coalgebraic characterisation of combinatorial parameters such as treewidth~\cite{AbramskyDawarWang} and pathwidth~\cite{pebblerelationgame}.
We can use Corollary \ref{co:coalgebras} to extract the same results.
As an example, consider the width of a tree-decomposition from $\cc_G$, where $G$ is the $k$-pebble game.
Condition $(ii)'$ 
implies that $\mathcal A$ has a $k$-pebble forest cover.
   From the fact that tree-decomposition of width ${<k}$ is equivalent to $k$-pebble forest cover \cite{relatingstructure}, we obtain a coalgebraic characterisation of treewidth similarly to \cite{AbramskyDawarWang}.
   Theorem \ref{thm:coalgs} and its corollary 
   open up the means to analogous investigations based on properties of event structures---their width, concurrency width, degree of confusion and causal depth~\cite{thesis}.    
For instance,
there are variations of Example~\ref{ex:multigph-revisited}, \eg~where in $G$ we have $k$ concurrent assignments and compare two undirected graphs. 
In this case, the coalgebras 
are colourings of 
graphs with chromatic number $\leq k$.

To our knowledge, all the game comonads mentioned in the literature, 
based on instantiations of moves as single elements of relational structures, fit within the scheme above. As do new game comonads,  such as that associated with the game for trace-inclusion, Example~\ref{traceinclusion};
as well as those exploiting the explicit concurrency of event structures. 
Note, not all 
categories of  coalgebras we obtain are {\em arboreal categories}:  that of Example~\ref{ex:multigph-revisited} will not be;  its coalgebras can take a non-treelike 
shape.
\setcounter{subsection}{3}
\subsubsection{
Characterisation of one-sided strategies.
}\label{sec:onesidedstrategies}

Even in the one-sided case, not all strategies in $\SD_\delta$ are coKleisli maps of game comonads.
Here we characterise all $\SD_\delta$ in the one-sided case.

Consider now a strategy $(\sig,\rho): \A\profto_\delta \B$ in $\SD_\delta$ where $\A$ and $\B$ are $\Sigma$-structures. In this one-sided case the situation of Section~\ref{sec:partlexpn} simplifies. Again,
 $\sig$ factors through 
the expansion of $G^\perp\vvbar G$ with respect to $\A$ and $\B$ as 
\[\sig: S\arr{\sig_0} \expn(G^\perp, \A)\vvbar \expn(G, \B) \to G^\perp\vvbar G.\]
The expansion 
w.r.t.~$\A$ and $\B$ factors into the expansion  w.r.t.~$\A$ followed by the expansion w.r.t.~$\B$ as
$$\expn(G^\perp, \A)\vvbar \expn(G, \B) \to \expn(G^\perp,\A)\vvbar G \to G^\perp\vvbar G.$$
Now we obtain the commuting diagram
\[
\vcenter{\xymatrix{
S\ar[d]_{\sig_0}\ar[rr]\ar@{..>}[rd]^\theta&&D\ar[d]^\delta\\
\expn(G^\perp, \A)\vvbar \expn(G, \B)\ar[rd]&\DA\pb{-45}\ar[d]_{\pi_2} \ar[ru]^{\pi_1}&G^\perp\vvbar G\\
&\expn(G^\perp,\A)\vvbar G\ar[ru]_\epsilon&
}}
\]
where the dotted arrow represents the unique mediating map  $\theta$ to the pullback. As a direct corollary of Lemma~\ref{cor:one-wayiso}
we get:

\begin{corollary}\label{cor:one-wayiso}
The map  $\theta$ 
is an isomorphism. 
\end{corollary}

Directly from $\theta$ being an isomorphism
, we can reformulate the one-sided strategies of $\SD_\delta$.

\begin{theorem}
Strategies $(\sig_0,\rho_0)\in\SD_\delta(\A, \B)$ bijectively correspond 
to
strategies $(\delta(\A),\rho(\A) \cup h
)
$,  where 
$h: \mod{\DA}_{V_2} \to \mod{\B}$ is a sort-respecting function, $\rho(\A) \cup h$ acts as $\rho(\A)$ on $V_1$-moves of $\DA$ and as $h$ on $V_2$-moves, and
for all +-maximal $x\in\iconf\DA$,
$x\models_{\delta(\A), \rho(\A) \cup h
} W_1\to W_2.$  

(As earlier, $W_1$ and $W_2$ are copies of the winning condition $W$ of $G$ with renamings by constants and variables of respectively the left and right of $(\Sigma+\Sigma, C+C, V+V)$.)
\end{theorem}
The theorem makes precise how one-sided strategies of $\SD_\delta$ are determined by the extra part of the instantiation to $V_2$-moves, given by $h$ above.

%
Strategies $\SD_\delta(\A, \B)$ bijectively correspond to sort-respecting  functions $h: \mod{\DA}_{V_2} \to \mod\B$  such that for each $+$-maximal configuration $x\in\iconf\DA$ we have:
\begin{equation}\label{hchar}
 x_{G^\perp}\models W \text{ implies } (x[h])_G \models W.
\end{equation}
Above,  $x[h]$ is the configuration of $\expn(\DA, \B)$ in which 
 the $V_2$-moves of $x$ are instantiated to the elements of $\B$ prescribed by $h$. We also use the fact that $x$  projects to  a configuration  $x_{G^\perp}$ of 
 ${\expn(G^\perp,\A)}$ and $x[h]$ projects to a configuration $(x[h])_G$ of  
 ${\expn(G,\B)}$.
%

We call a function $h$ satisfying (\ref{hchar}) a {\em $\delta$-homomorphism} from $\A$ to $\B$. 
We turn this characterisation into a direct presentation of 
$\SD_\delta$ 
as relational structures with 
$\delta$-homomorphisms. 
Recall from Section~\ref{sec:SDconstrn} that the candidates for identity strategies, \viz~$(\iota_\A, \rho_\A):\A\profto_\delta \A$, only become so when they are winning.  
We say that {\em $\SD_\delta$ has identities} when the strategies $(\iota_\A, \rho_\A)$ are winning for each 
relational structures $\A$. 

\begin{toappendix}
    \subsection{One-sided characterisation}
\end{toappendix}

\begin{theoremrep}[one-sided characterisation] Let $\A$, $\B$ and $\C$ be $\Sigma$-structures.  
\begin{enumerate}
\item\label{sdcon1} Assuming  $\SD_\delta$ has identities, 
the function $j_\A: \mod\DA_{V_2}\to\A$ such that  $j_\A(e) = \rho(\A)(\bar e)$ is a $\delta$-homomorphism.
\item\label{sdcon2} Suppose $h$ is a $\delta$-homomorphism from $\A$ to $\B$, and  $k$ is a $\delta$-homomorphism from $\B$ to $\C$. Then, 
the composition $k\circ h^\dagger:  \mod{\DA}_{V_2} \to  \mod\C$ is a $\delta$-homomorphism from $\A$ to $\C$.
\end{enumerate}
Further, $\Sigma$-structures and $\delta$-homomorphisms form a category with identities given by (\ref{sdcon1}) and composition given by (\ref{sdcon2}).
Provided  $\SD_\delta$ has identities,  the category $\SD_\delta$ is isomorphic to the category of $\Sigma$-structures with $\delta$-homomorphisms. 
\end{theoremrep}
\begin{proof}
(i) Assume  $\SD_\delta$ has identities, \ie~the strategy $(i_\A
, \rho_\A
)$ of Section~\ref{sec:SDconstrn} is winning so an identity strategy in $\SD_\delta$. From Lemma~\ref{cor:one-wayiso}, the identity strategy $(\iota_\A, \rho_\A)$ is isomorphic to the strategy $(\iota_\A(\A), 
\rho(\A)+j_\A)$.  
  Because $(\iota_\A, \rho_\A)$ is winning, so is  $(\iota_\A(\A), 
\rho(\A)+j_\A)$.  Thus if $x$ is a +-maximal configuration of $\DA$ such that $x_{G^\perp} \models W_G$, then 
$(x[j_\A])_{G} \models W_G$, as required for $j_\A$ to be a $\delta$-homomorphism.

\noindent
(ii)
From (i), when $\SD_\delta$ has identities $(\iota_\A, \rho_\A)$ they correspond with the $\delta$-homomorphisms $j_\A$. 
Let $\A$, $\B$ and $\C$ be $\Sigma$-structures. Let $h$ be a $\delta$-homomorphism from $\A$ to $\B$, and  $k$ a $\delta$-homomorphism from $\B$ to $\C$; they correspond to strategies $\A\profto_\delta \B$ and $\B\profto_\delta \C$ . 
We need to check the composition $k\circ h^\dagger:  \mod{\DA}_{V_2} \to  \mod\C$ is identical to the $\delta$-homomorphism, say $l$,  corresponding to the composition $\A\profto_\delta \C$ of the two strategies.  

The $\delta$-homomorphism $h$ from $\A$ to $\B$ corresponds to a concurrent strategy
$$
\sig_0: \DA \to \expn(G^\perp, \A) \vvbar \expn(G, \B)\,
$$
while $k$ from $\B$ to $\C$ 
corresponds to a concurrent strategy
$$
\sig_0: \DB \to \expn(G^\perp, \B) \vvbar \expn(G, \C).
$$
It is convenient to identify the configurations of $\DA$ with their images in $\expn(\DA, \B)$ given by the instantiation $h$ : 
a configuration then has the 
form $(y[\rho_A])[h]$ for some $y\in\iconf D$ and sort-respecting $\rho_A:\mod D_{V_1} \to \mod\A$.  Similarly, a configuration of the concurrent strategy for $k$ has the 
form $(z[\rho_B])[k]$ for some $z\in\iconf D$ and sort-respecting $\rho_B:\mod D_{V_1} \to \mod\B$.  

Consider an arbitrary $(e_0,\ga_0)\in \mod\DA_{V_2}$.  There exist $x\in\iconf D$ and $\rho_A: \mod D_{V_1} \to \mod\A$ such that 
 $(e_0,\ga_0)\in x[\rho_A]$; then $e_0\in x$ and $\ga_0= \rho_A(e') $ for all  $V_1$-moves $e'\leq_D e_0$.  
 
  From the idempotency of $\delta$ there are $y, z\in\iconf D$ for which $d x = z\scirc y$.  
 Consider a specific configuration of the interaction of the two strategies, \viz
 $$((z[\rho_B])[k])\sncirc (y[\rho_A])[h]),$$
 where $(y[\rho_A])[h]\in \expn(\DA, \B)$ and $(z[\rho_B])[k]\in  \expn(\DB, \C)$;
 in order for the interaction to be defined, for all $V_1$-moves $e_1\in z$, we must have
$$\rho_B(e_1) = h((\bar e_1, \ga_1))$$
 with 
 $
 \ga_1(e') = \rho_A(e')   \hbox{ for all $V_1$-moves } e'\leq_D \bar e_1.
 $
  
 Consider a $V_2$-move $(e,\ga)\in z[\rho_B]$.  Then,
  $$
 \ga(e_1) = \rho_B(e_1)  \hbox{ for all $V_1$-moves } e_1\leq_D e, \ \ie
 $$
$$
 \ga(e_1) = h(\bar e_1, \ga_1) \hbox{ for all $V_1$-moves } e_1\leq_D e,\hbox{ \quad } 
 $$
 where $\ga_1$ is defined above.  In other words, recalling the definition of $h^\dagger$, 
 $$(e,\ga) = h^\dagger((e,\ga')),$$
 where $
 \ga'(e') = \rho_A(e')   \hbox{ for all } e'\leq_D e \hbox{ with $e'$ a $V_1$-move.} 
 $
 Consequently, the $V_2$-move $(e,\ga)\in z[\rho_B]$ is instantiated to $k(h^\dagger((e,\ga')))$  in $(z[\rho_B])[k]$.  
Hence $k\circ h^\dagger((e, \ga')) = l((e, \ga'))$.  

In particular, there is a 
$V_2$-move $(e,\ga)\in z[\rho_B]$ for which   
 $(e, \ga) = h^\dagger((e,\ga_0)))$, ensuring $k\circ h^\dagger((e, \ga_0)) = l((e, \ga_0))$ for an arbitrary $(e_0,\ga_0)\in \mod\DA_{V_2}$. Thus $l$, the $\delta$-homomorphism corresponding to the composition of strategies, and $k\circ h^\dagger$ are the same functions from  $\mod\DA_{V_2}$ to $\mod\C$.
 \end{proof}

The coKleisli categories of the game comonads of Section~\ref{sec:comonads}, coincide with 
special Spoiler-Duplicator categories $\SD_\delta$: 

\begin{theoremrep}
Assume $G$ has homomorphic winning condition and $\delta$ is balanced. Then $\SD_\delta$ has identities. W.r.t.~the comonad $\Rel_\delta(\_)$, any coKleisli map 
is a $\delta$-homomorphism
; the coKleisli maps are precisely those $\delta$-homomorphisms which are also $\Sigma$-homomorphisms. Consequently, 
$\SD_\delta$ is the coKleisli category of the comonad $\Rel_\delta(\_)$ iff every $\delta$-homomorphism is a $\Sigma$-homomorphism.
\end{theoremrep}
\begin{proof}
Recall the comonad $\delta:D\to G^\perp\vvbar G$ 
and that $G$ has the homomorphic winning condition $W$.  

Let $\B$ be a $\Sigma$-structure.  Via Corollary~\ref{cor:one-wayiso}, its putative identity in $\SD_\delta$ is given by the instantiation with map $\delta(\B):\D(\B)\to G^\perp\vvbar G$ and the function
$\rho_\B: \mod{\D(\B)}_{V_2} \to \mod\B$, taking $(e,\ga)$ to $\ga(\bar e)$. 
Let $x$ be a +-maximal configuration of $\D(\B)$ with $x_{G^\perp}\models W$.  Then, using that $\delta$ is balanced, $(x[\rho_\B])_G \models W$. 
It follows that 
the putative identity is winning, and is indeed the identity at $\B$ in $\SD_\delta$.  

We show that coKleisli maps are $\delta$-homomorphisms, from which the remaining claims follow. 
Let $h:\Rel_\delta(\A) \to \B$ be a coKleisli map, so   a function $h:\mod{D(\A)}_{V_2} \to \mod B$ preserving the relational structure. Let $x$ be a +-maximal configuration of $D(\A)$ such that $x_{G^\perp} \models W$.  
The map 
$\pi_2: D(\A) \to \expn(G^\perp,\A)\vvbar G$ sends $x$ to $x_{G^\perp}\vvbar x_G$ where $x_{G^\perp}\in\iconf{\expn(G^\perp,\A)}$ and $x_G\in\iconf G$. As $\delta$ is balanced, $x_{G^\perp}$ has $x_G$ as image in $G^\perp$. 

There is a natural instantiation of $G$ in $\Rel_\delta(\A)$ determined by $x_G$: the map of the instantiation is the inclusion $j:x_G\hookrightarrow G$, regarding $x_G$ as an event structure got as the restriction of $G$; the function of the instantiation is $\rho:  \mod{x_G}_V\to \mod{\Rel_\delta(\A)}$ given by the composite $ \mod{x_G}_V\iso \mod x_{V_2} \hookrightarrow \mod{\Rel_\delta(\A)}$.
Structural induction shows that for any assertion $\phi$ of the free logic for $G$,
$$
x_{G^\perp}\models \phi \  \hbox{ iff }\  x_{G}\models_{j,\rho} \phi
$$
---note that on the left assertions are interpreted w.r.t.~$\A$, and on the right w.r.t.~$\Rel_\delta(\A)$.
Now,
$$
\eqalign{
x_{G^\perp}\models W &\  \iff \  x_{G}\models_{j,\rho} W\,, \hbox{ by the above,}\cr
&\implies x_{G}\models_{j,h\circ \rho} W\,,\hbox{ as $W$ is homomorphic,}\cr
&\iff (x[h])_{G}\models W\,,}
$$
as required for $h$ to be a $\delta$-homomorphism.
\end{proof}

The role of coKleisli maps in~\cite{AbramskyDawarWang} is replaced by the more general notion of $\delta$-homomorphisms. Categories $\SD_\delta$ are strictly more general than coKleisli categories of game comonads: $G$ need not have homomorphic winning condition, \eg~if its winning condition involves negation; even when $G$ has homomorphic winning condition, a $\delta$-homomorphism need not be a $\Sigma$-homomorphism.
In replacing the role of a map in a coKleisli category of a game comonad, it is now the presence of a $\delta$-homomorphism which is equivalent to truth preservation in a logic characterised by its Spoiler-Duplicator game; for instance, $k$-variable logic in the case of the $k$-pebble game \cite{AbramskyDawarWang} and its restricted conjunction fragment in the case of the all-in-one game \cite{pebblerelationgame} ---\cf~Section~\ref{sec:finrem}.

\section{Two-sided games}
Theorem~\ref{thm:two-sided} expresses how a two-sided strategy in a Spoiler-Duplicator game reduces to a pair of functions. It entails
:

\begin{corollary}\label{cor:twosided}
Strategies $\SD_\delta(\A, \B)$ bijectively correspond to pairs  of sort-respecting functions
$$
k: \mod\DAB_{V_1^+} \to \mod\A
\hbox{ and }
 h: \mod\DAB_{V_2^+} \to \mod\B\,
 $$
 such that for all +-maximal $x\in\iconf\DAB$,
 $$
(x[k])_{G^\perp}\models W \implies (x[h])_G \models W.  
$$
\end{corollary}

Does the representation in terms of functions enable us to 
express the composition of two-sided Spoiler-Duplicator strategies more simply, as happened in the one-sided case? In general, there is a description of composition in the intricate manner of Geometry of Interaction~\cite{lics23}.  There is however a simpler description when the comonad $\delta$ is copycat.
When $\delta= \cc_G$, Spoiler-Duplicator strategies 
correspond to {\em $G$-spans} of event structures, and their composition is 
simplified to the standard composition of spans via pullbacks.

Throughout this section, we adopt the notation $G^+$ for the event structure obtained as 
the projection of an 
esp $G$ to its Player moves.
The conversion of strategies to spans 
hinges on the map
$$\Sq: \CC_G^+ \to G,$$
given by taking $\Sq (1,g) = \bar g$ and 
$\Sq (2,g) = g$. It is an  isomorphism of event structures which folds 
 the Player moves on the left of $\CC_G$ to their corresponding Opponent moves on the right. Precomposing $F$ with 
 the projection $\CC_G(\A,\B)^+\to \CC_G^+$ associated with the partial expansion, we define the map of event structures 
$$\sq_{\A,\B}: \CC_G(\A,\B)^+ \to G.$$

We construct the $G$-span of event structures corresponding to a strategy 
in $\SD_{\cc_G}(\A, \B)$ given by 
functions 
$h$ and $k$.   
Its vertex, $\CC_G(\A,\B)^+$  is the  projection of the partial expansion $\CC_G(\A,\B)$ to its Player moves.  Its two maps are got by extending $k$ and $h$ of Corollary~\ref{cor:twosided} to maps of event structures.  First 
define functions
$$
K: \mod{\CC_G(\A,\B)^+}_{V+V} \to \mod\A \hbox{ and } 
H: \mod{\CC_G(\A,\B)^+}_{V+V} \to \mod\B,
$$
both from the variable-labelled Player moves of $\CC_G(\A,\B)$
,
by
$$
K(s) = 
\begin{cases} k(s) & \hbox{ if } s\in \mod{\CC_G(\A,\B)}_{V_1^+},\\
 \rho(\A) (\bar s)  & \hbox{ if }  s\in \mod{\CC_G(\A,\B)}_{V_2^+}
\end{cases}
$$
and
$$
H(s) = 
\begin{cases} h(s) & \hbox{ if } s\in \mod{\CC_G(\A,\B)}_{V_2^+},\\
 \rho(\B) (\bar s)  & \hbox{ if }  s\in \mod{\CC_G(\A,\B)}_{V_1^+}.
\end{cases}
$$
Every $s\in \mod{\CC_G(\A,\B)^+}_{V+V}$ is now associated with a pair of values 
$(K(s), H(s))$ in the product of $\Sigma$-structures $\A\times \B$.  

From the map $\sq_{\A,\B}$ and the functions $K$ and $H$, we obtain an instantiation
$(\sq_{\A,\B}, K)$ of $G$ in $\A$, and an instantiation $(\sq_{\A,\B}, H)$  of  $G$ in $\B$. By  the universality of the expansions of $G$ w.r.t.~$\A$ and $\B$ ---Lemma~\ref{lem:universalinst},
we get maps of event structures 
$$
\xymatrix@R=10pt@C=5pt{
&\ar[dl]_{l_\A} \CC_G(\A,\B)^+\ar[dr]^{r_\B}&\\
\expn(G, \A)&& \expn(G,\B)
}
$$
---the $G$-span corresponding to the original strategy in $\SD_{\cc_G}(\A,\B)$. 
From the original strategy 
being winning we derive: 
 if $l_\A x \models W$ then $r_\B x \models W$, for any $x\in \iconf{\CC_G(\A,\B)^+}$.  
\begin{toappendix}
We explain in more detail how $G$-spans correspond to two-sided strategies in $\SD_{\cc_G}$. By Corollary~\ref{cor:twosided} such a strategy corresponds to a pair of functions $h$ and $k$ from which we construct a span
$$
\xymatrix@R=10pt@C=5pt{
&\ar[dl]_{l_\A} \CC_G(\A,\B)^+\ar[dr]^{r_\B}&\\
\expn(G, \A)&& \expn(G,\B)\,.
}
$$
The span is ``winning'' in the sense that  $l_\A x \models W$ implies $r_\B x \models W$, for any $x\in \iconf{\CC_G(\A,\B)^+}$.  
By the universality of expansions involved in their definitions,
the maps $l_\A$ and $r_\B$ are unique such that
$$
\epsilon_{G,\A}\circ l_\A=\sq_{\A,\B}
\quad \hbox{ and }\quad 
\rho_{G,\A}\circ l_\A(s) =\rho(\A)(\bar s)\,, \hbox{ for } s\in\mod{\CC_G(\A,\B)}_{V_2^+} \eqno(a)
$$
and
$$
\epsilon_{G,\B}\circ r_\B=\sq_{\A,\B}
\quad \hbox{ and }\quad 
\rho_{G,\B}\circ r_\B(s) =\rho(\B)(\bar s)\,, \hbox{ for } s\in\mod{\CC_G(\A,\B)}_{V_1^+}\,. \eqno(b)
$$
The original functions $h$ and $k$ are recovered as
$$
k(s) =\rho_{G,\A}\circ l_\A (s)\,, \hbox{ for } s\in \mod{\CC_G(\A,\B)}_{V_1^+}\,,
$$
and
$$
h(s) =\rho_{G,\B}\circ r_\B (s)\,, \hbox{ for } s\in \mod{\CC_G(\A,\B)}_{V_2^+}\,.
$$
Thus winning $G$-spans $l_\A$, $r_\B$, satisfying ($a$) and ($b$), 
bijectively correspond to 
two-sided strategies in $\SD_{\cc_G}$.  
\end{toappendix}

Two $G$-spans compose via pullback to form a $G$-span
:
$$
\xymatrix@R=10pt@C=0pt{&&\ar@{..>}[dl]\CC_G(\A,\C)^+\ar@{..>}[dr]\pb{270}&&\\
&\ar[dl]_{l_\A} \CC_G(\A,\B)^+\ar[dr]^{r_\B}& &\ar[dl]_{l_\B} \CC_G(\B,\C)^+\ar[dr]^{r_\C}&\\
\expn(G, \A)&& \expn(G,\B)&& \expn(G,\C)
}
$$ 

\begin{theoremrep}
$\SD_{\cc_G}$ is isomorphic to the category of $G$-spans.
\end{theoremrep}
\begin{proof} 
(Outline) 
A strategy in $\SD_{\CC_G}(\A,\B)$ takes the form of an instantiation $(\sig(\A,\B), \rho_{\A,\B})$ where recall $\sig(\A,\B):\CC_G(\A,\B)\to G^\perp\vvbar G$ and $\rho_{\A,\B}:\mod{\CC_G(\A,\B)^+}_{V+V}\to \mod{\A+
\B}$.  Similarly, a strategy in $\SD_{\CC_G}(\B,\C)$ is an instantiation $(\sig(\B,\C), \rho_{\B,\C})$.   
From Lemma~\ref{lem:two-wayiso}, in their composition,  
$$ \CC_G(\A, \C) \iso  \CC_G(\B,\C)\scirc \CC_G(\A, \B)\,.$$

 Furthermore,  
$$(\CC_G(\B,\C)\scirc \CC_G(\A, \B))^+\iso 
\CC_G(\A, \B)^+ \wedge \CC_G(\B,\C)^+\,,
$$
the associated pullback in the composition of spans.  To see this, consider a configuration $w$ of the lhs. Its downclosure $[w]$ in the interaction $\CC_G(\B,\C)\sncirc \CC_G(\A, \B)$,  
of the form $y\sncirc x$, is such that for both $x$ and $y$ their images in  $G^\perp\vvbar G$,
$$
\sig(\A,\B)\, x = z\vvbar z \ \hbox{ and }\ \sig(\B,\C)\, y  = z\vvbar z\,,
$$
for some $z\in\iconf G$ ---this is because $y\sncirc x$ is the downclosure of +-moves;
with the instantiations $\rho_{\A,\B}$ and $\rho_{\B,\C}$ agreeing over $\B$.
From this, 
the arbitrary configuration $w$ of the lhs bijectively corresponds to a configuration $x^+\wedge y^+$ of the pullback on the rhs.

This shows the agreement between the composition of spans and that of the strategies with which they correspond.   
\end{proof}

We were looking for a characterisation, so correspondence, with (deterministic) delta-strategies, and could only get that via spans when delta is copycat. Bisimulations describe a form of nondeterministic strategy---they do not fix Player to a single answering move---and this is why there is only an equivalence, and not a correspondence, between bisimulations and deterministic strategies in~\cite{relatingstructure}; and why bisimulations do not figure more centrally in our account.

\section{Final remarks}\label{sec:finrem}
We have introduced a general method for generating comonads from one-sided Spoiler-Duplicator games, provided a characterisation of strategies for one-sided games in general, delineated when they reduce to coKleisli maps, and captured an important subclass of two-sided strategies as spans. 

As a 
next step, we will focus on generalising the framework to allow moves assigning  subsets, rather than individual elements of the relational structure---extending the free logic to allow second order variables and quantifiers in winning conditions.
This 
will bring \emph{bijective games} within the framework 
\cite{AbramskyDawarWang,relatingstructure}. 
Their interest 
is illustrated through the bijective $k$-pebble game which corresponds to indistinguishability by the $k$-dimensional Weisfeiler-Lehman algorithm~\cite{immerman}.
Extending moves to assign subsets is  
a crucial avenue for future work, making other games accessible to structural, categorical methods. 
Notable 
challenges are games characterising (fragments of) Monadic Second Order Logic and $\mu$-calculus. We are optimistic because of the
adaptable nature of the games here. 

A comonad $\delta$ in signature games specifies Spoiler-Duplicator strategies $\SD_\delta$.
 Imagine a logic extending the free logic. Within it a {\em $\delta$-assertion} is an assertion  preserved by all strategies in $\SD_\delta$; it is preserved by a strategy  in the same manner as   the winning condition.
 A task is that of  finding a structural characterisation of the $\delta$-assertions  corresponding to   $\delta$, and {\it vice versa}; how tuning the rich array of combinatorial parameters encapsulated in $\delta$ correlates with changes in the structural properties of $\delta$-assertions. We can potentially study how such correlations vary with $\delta$, using 
 the fact that such comonads themselves form a category.

Finally, since strategies are operational refinements of presheaves~\cite{fossacs13} they could potentially inform Lovász-type theorems, as resource-aware refinements of the Yoneda Lemma. 
Another advantage of concurrent games is that they already
accommodate extensions, 
 for instance 
to games with imperfect information~\cite{Dexter}, and 
enrichments with probability or quantum effects~\cite{Probstrats,POPL19,concquantumstrats}, appropriate  to the study of probabilistic and quantum advantage.

 \section{Addendum: Free logic}\label{add:freelogic}
 A  winning condition for a $(\Sigma, C, V)$-game is an assertion in the free logic over $(\Sigma, C, V)$. 
 We explain the syntax and semantics of the free logic. 
 A good reference for free logic is Dana Scott's article~\cite{Scott-freelogic}.   

The free logic uses an explicit existence predicate $\EX(t)$ to specify when a term $t$ denotes something existent. The existence of elements is highly relevant when exploring relational structures with bounded resources as in pebbling games where the number of pebbles bounds the existent part of the relational structure, so what assertions it satisfies, at any stage. 

Throughout this section, assume a  \emph{$(\Sigma, C, V)$-game over a $\Sigma$-structure}  $(G, \mathcal A)$, comprising a $(\Sigma, C, V)$-game $G$ and a  $\Sigma$-structure $\mathcal A$. 
In providing the semantics to the free logic it will be convenient  to name  the elements of the $\Sigma$-structure $\A$ as terms.  
{\em Terms} of the free logic are either variables, constants or  elements of the relational structure: 
\begin{align*}
    t::=\alpha\in V\ \mid\  c\in C\ \mid\ a\in |\mathcal{A}| 
\end{align*}

\noindent
{\em Assertions} of the free logic are given by:
\begin{align*}
    \phi::=&\ R(t_1,\ldots,t_k)\ \mid\ t_1=t_2\ \mid\ \EX(t) \ \mid\ t_1\prel t_2\mid\ t_1\pre t_2\\&
    \ \mid\ \phi_1\wedge \phi_2\ \mid\ \phi_1\vee \phi_2\ \mid\ \neg\phi\ \mid\ \forall\alpha.\phi\ \mid\ \exists\alpha.\phi
    \ \mid\ \bigwedge_{i\in I}\phi_i\ \mid\ \bigvee_{i\in I}\phi_i
\end{align*}
with $t_i$ and $t$ ranging over terms of correct sorts w.r.t.~the arity of relation symbols 
and the indexing sets $I$ are countable.  
The predicate $\EX(t)$ asserts the existence of the value denoted by term $t$.
An assertion $t_1\prel t_2$ is an {\em order constraint} between terms. 
It expresses that $t_1$ is a latest constant or assignment move on which $t_2$ depends 
w.r.t.~$G$.  The {\em equality constraint}  $t_1\pre t_2$ ensures in addition that the two terms denote equal values in $\A$. 
Equality constraints are used in the ``all-in-one game'', 
Example~\ref{traceinclusion}. 
We adopt the usual abbreviations. \eg~an implication $\phi \to \phi'$  abbreviates  $\neg 
\phi \vee \phi'$. 


The semantics of assertions  is given w.r.t.~an instantiation $(\sig,\rho)$ of $G$  in $\A$, with $\sig:S\to G$ and $\rho:\mod S_V\to \mod\A$.
Given an instantiation $(\sig,\rho)$, we define the semantics of terms by:
\begin{align*}
    &\llbracket\alpha\rrbracket_{(\sig,\rho)}=\{(x,s, a)\in\iconf{S}\times S\times \mod\A\mid  s\in\last_S(x) \ \& \\
    & \qquad \qquad \qquad \  \Var(\sig(s)) =\al \ \&\ \rho(s) = a\}\\
    &
   \llbracket c\rrbracket_{(\sig,\rho)} = 
\{(x,s,c)\in\iconf{S}\times S\times C \mid  s\in x\  \& \ \Var(\sigma(s)) = c \ \&  \\
&\qquad \qquad  \qquad
\all s’ \in x.\  \sigma(s) \leq_G \sigma(s’) \ \&\  \Var(\sigma (s’))\in C \ \&\ \\& \qquad \qquad \qquad\pol_S(s) =\pol_S(s’) 
\Rightarrow  s’= s \}\\
    &\llbracket a\rrbracket_{(\sig,\rho)} =\{(x,s,a)\in\iconf{S}\times S\times\mod\A \mid s\in \last_S(x)\ \&\ \rho(s) =a \}
\end{align*}
\noindent
A term is denoted by a set of triples $(x, s, v)$ consisting of a configuration $x$, an event $s\in\last_S(x)$, at which its value $v$ becomes existent.  Notice, in particular, that a value in the algebra, $a\in\mod A$, is only regarded as existing at a configuration $x$ if there is a last  move in $x$, \viz~some $s\in\last_S(x)$, which is instantiated to $a$: accordingly, the denotation of ${\EX(a)}$, $a\in\mod\A$, below will be 
$$\den{\EX(a)}_{(\sig,\rho)}= \set{x\in\iconf S}{a\in\rho\, \last_S(x)},$$ \ie~the set of configurations at which $a$ exists.   
We have chosen to say that a constant $c$ exists in a configuration if an event bearing it has occurred there and there are no other {\em constant events of the same polarity}  which causally depend on it in that configuration; for a constant to exist in a configuration, it has to be a latest (causally maximal) constant event there 
w.r.t.~all other constant events of the same polarity.  (This interpretation is critical in Examples~\ref{simulation},~\ref{traceinclusion}.)

Define the size of assertions as an ordinal:
    
\begin{itemize}
\item $size(R(t_1,\ldots,t_k))=
        size(t_1=t_2)=
size(\EX(t))  $
\

\hspace{7.34em} $ = size(t_1\prel t_2)= size(t_1\pre t_2)=1$; 

  \item $size(\neg\phi)=size(\exists\alpha.\phi)=size(\phi)+1$;
    \item $size(\phi_1\wedge\phi_2)=sup\{size(\phi_1),size(\phi_2)\}+1$;
    \item 
    $size(\bigwedge_{i\in I}\phi_i)=sup\{size(\phi_i)\mid i\in I\}+1$.
\end{itemize}

 \begin{definition}[semantics]
 Given an instantiation $(\sig,\rho)$, with $\sig:S\to G$, the semantics of an assertion
specifies the set of configurations of   $S$ which satisfy it.  
The key clauses of the semantics are defined below by induction on the 
size of assertions:

   \begin{itemize}
       \item $\llbracket R(t_1,\ldots,t_k)\rrbracket_{(\sig,\rho)}=$
       \begin{align*}
        \{x\in\iconf{S}\mid &\ \exists (a_1,\cdots,a_k)\in R_\mathcal{A}, s_1, \cdots, s_k\in S.\\&
        (x,s_1,a_1)\in\llbracket t_1\rrbracket_{(\sig,\rho)}\ \& \dots\ \&\ (x,s_k,a_k)\in \llbracket t_k\rrbracket_{(\sig,\rho)}\}
        \end{align*}

   \item $\llbracket t_1=t_2\rrbracket_{(\sig,\rho)}=$
   \begin{align*}
    \{x\in\iconf{S}\mid & 
    \ \exists a\in\mod A, s_1, s_2\in S.\\& \ (x,s_1,a)\in\llbracket t_1\rrbracket_{(\sig,\rho)} \ \&\ (x,s_2,a)\in \llbracket t_2\rrbracket_{(\sig,\rho)}\}
    \end{align*}
    
   \item$ \llbracket\EX(t)\rrbracket_{(\sig,\rho)}=$
   \begin{align*}
       \set{x\in\iconf S}{\exists v\in \mod \A\cup C, s\in S.\ (x, s, v) \in\llbracket t\rrbracket_{(\sig,\rho)}}
   \end{align*}
     
   \item $\llbracket t_1\prel t_2\rrbracket_{(\sig,\rho)}=$
    \begin{align*}
         \{x\in\iconf S \mid &\ 
          \exists 
   s_2, s_1\in S, v_1, v_2\in \mod\A\cup C.\
   \sig(s_1)<_G \sig(s_2)\ \&\  \\ 
   &
   (x, s_2, v_2) \in \den{t_2}_{(\sig,\rho)} \ \&\ 
   ([s_2]_S, s_1, v_1)  \in \den{t_1}_{(\sig,\rho)}    
   \}
    \end{align*}

      \item $\llbracket t_1\pre t_2\rrbracket_{(\sig,\rho)}=$
    \begin{align*}
         \{x\in\iconf S \mid &\ 
          \exists 
   s_2, s_1\in S, a\in \mod\A.\
   \sig(s_1)<_G \sig(s_2)\ \&\  \\ 
   &
   (x, s_2, a) \in \den{t_2}_{(\sig,\rho)} \ \&\ 
   ([s_2]_S, s_1, a)  \in \den{t_1}_{(\sig,\rho)}    
   \}
    \end{align*}
 
   \item $\llbracket\phi_1\wedge\phi_2\rrbracket_{(\sig,\rho)}=\llbracket\phi_1\rrbracket_{(\sig,\rho)}\cap\llbracket\phi_2\rrbracket_{(\sig,\rho)}$
  \item   $\llbracket\neg\phi\rrbracket_{(\sig,\rho)}=\iconf{S}\setminus\llbracket\phi\rrbracket_{(\sig,\rho)}$

   \item $\llbracket\exists\alpha.\phi\rrbracket_{(\sig,\rho)}=$
   \begin{align*}
      \{x\in\iconf{S}\mid&\ \exists a\in|\mathcal{A}|. sort(a)=sort(\alpha)\ \&\ \\&  x\in \den{\EX(a)}_{(\sig,\rho)} 
    \ \&\  x\in\llbracket\phi[a/\alpha]\rrbracket_{(\sig,\rho)}\} 
   \end{align*}
    
  \item $  \llbracket\bigwedge_{i\in I}\phi_i\rrbracket_{(\sig,\rho)}=\bigcap_{i\in I}\llbracket\phi_i\rrbracket_{(\sig,\rho)}$
  
  \end{itemize}

 \end{definition}
 We write $x\models_{\sig,\rho}
\phi$ iff $x\in\llbracket\phi\rrbracket_{(\sig,\rho)}$.  
Note an assertion $R(t_1,\ldots, t_k)$ is only true at a configuration at which all the terms $t_1,\ldots, t_k$ denote existent elements. 

\begin{acks}
Thanks to  Aurore Alcolei, 
Pierre Clairambault, 
Adam Ó Conghaile,
Mai Gehrke, 
Sacha Huriot-Tattegrain and
Martin Hyland, as well as to the anonymous referees.  As part of his Cambridge internship from ENS Saclay, Sacha verified the key facts about expansions~\cite{sacha}.  Glynn Winskel was partially supported by the Huawei Research Centre, Edinburgh.
\end{acks}
 
\bibliography{biblio}
\end{document}